\documentclass[11pt]{article}

\usepackage{arxiv}
\usepackage{cite}
\usepackage{amsmath,amssymb,amsfonts}
\usepackage{amsthm,mathtools}
\usepackage{bm}
\usepackage{graphicx}
\usepackage{textcomp}
\usepackage{xcolor}
\usepackage{braket}
\usepackage{array,booktabs}
\usepackage{subcaption}
\usepackage{algorithm}
\usepackage{algpseudocode}
\usepackage{xurl}
\usepackage{microtype}
\usepackage{placeins}
\usepackage{caption}
\usepackage{float}
\usepackage{pifont}
\usepackage{tikz}
\usetikzlibrary{shapes.geometric, arrows, positioning, fit, backgrounds}
\usepackage[hidelinks]{hyperref}

\newtheorem{theorem}{Theorem}

\newtheorem{corollary}{Corollary}
\newtheorem{definition}{Definition}
\newtheorem{example}{Example}
\newtheorem{remark}{Remark}

\DeclarePairedDelimiter{\norm}{\lVert}{\rVert}
\DeclarePairedDelimiter{\abs}{\lvert}{\rvert}
\renewcommand{\norm}[1]{\norm*{#1}}
\renewcommand{\abs}[1]{\abs*{#1}}

\DeclareGraphicsExtensions{.pdf,.png,.jpg,.eps}
\graphicspath{{figuras/}{./}}

\newcommand{\safeincludegraphics}[2][]{%
	\IfFileExists{#2}{%
		\includegraphics[#1]{#2}%
	}{%
		\begingroup
		\setlength\fboxsep{2pt}%
		\color{red}\fbox{\texttt{missing: #2}}%
		\endgroup
	}%
}

\providecommand{\adj}{\dagger}

\title{Digital Twin--Driven Adaptive Wavelet Strategy for Efficient 6G Backbone Network Telemetry}

\author{
	Alexandre Barbosa de Lima\\
	Pontifical Catholic University of S\~ao Paulo, Brazil\\
	\texttt{ablima@pucsp.br}
	\and
	Xavier Hesselbach\\
	Universitat Polit\`ecnica de Catalunya, Spain\\
	\texttt{xavier.hesselbach@upc.edu}
	\and
	Jos\'e Roberto de Almeida Amazonas\\
	University of S\~ao Paulo, Brazil\\
	\texttt{jose.amazonas@usp.br}
}

\date{}

\begin{document}
	
\maketitle
	
\begin{abstract}
Classical orthogonal wavelets guarantee perfect reconstruction but rely on fixed bases optimized for polynomial smoothness, achieving suboptimal compression on signals with fractal spectral signatures. Conversely, learned methods offer adaptivity but typically enforce orthogonality via soft penalties, sacrificing structural guarantees.
		
This work establishes a rigorous equivalence between Multiscale Entanglement Renormalization Ansatz (MERA) tensor networks and paraunitary filter banks. The resulting framework learns adaptive wavelets while enforcing exact orthogonality through manifold-constrained optimization, guaranteeing perfect reconstruction and energy conservation throughout training.
		
Validation on Long-Range Dependent (LRD) network traffic demonstrates that learned filters outperform classical wavelets by 0.5--3.8~dB PSNR on six MAWI backbone traces (2020--2025, 314~Mbps--1.75~Gbps) while preserving the Hurst exponent within estimation uncertainty ($|\Delta H| \le 0.03$). These results establish MERA-inspired wavelets as a principled approach for telemetry compression in 6G digital twin synchronization.
\end{abstract}
	
	
		
\section{Introduction}
\label{sec:intro}
Network traffic in future 6G systems is expected to exhibit long-range dependence (LRD), characterized by power-law correlation decay: $\phi(k) \sim k^{-\beta}$, $0 < \beta < 1$~\cite{Leland1994,Paxson1994,Millan2021LRD}. This fractal structure has profound implications for network management: buffer overflow probabilities decay polynomially  -- not exponentially -- with buffer size, fundamentally challenging classical queueing models~\cite{Norros1995,Parulekar1996}. For digital twins (DT) driving closed-loop optimization, capturing these self-similar dynamics is not merely a statistical exercise but a prerequisite for stability. As highlighted in surveys on machine learning for networking~\cite{Boutaba2018,Aouedi2025DLTraffic}, these dynamics call for adaptive multiscale representations capable of capturing traffic correlations across scales while preserving physical interpretability and robustness.

This work focuses on backbone aggregation telemetry, where traffic from thousands of edge cells converges into high-capacity core links. Statistical aggregation theorems establish that superposition of heterogeneous sources with heavy-tailed distributions preserves or amplifies LRD at such aggregation scales~\cite{taqqu1997proof,willinger1998self}, making backbone traces a natural testbed for validating LRD-preserving transforms. This setting motivates the development of adaptive wavelets that exploit traffic-specific correlation structures while maintaining the mathematical guarantees (perfect reconstruction, energy conservation) required for reliable signal processing. While 6G networks will encompass heterogeneous access technologies -- from millimeter-wave massive MIMO to satellite non-terrestrial networks -- wireless edge telemetry, with its distinct statistical properties induced by channel fading and mobility, represents a complementary challenge identified as future work (Section~\ref{sec:conclusion}).

While the focus of this study is on backbone aggregation telemetry rather than access-level radio traffic, this choice is deliberate. Digital twin synchronization fundamentally depends on preserving the statistical invariants of aggregated traffic flows -- most notably LRD -- which arise from multiplexing heterogeneous sources and persist independently of the underlying access technology (4G, 5G, or beyond). As such, backbone telemetry provides a technology-agnostic and representative testbed for validating synchronization-preserving compression mechanisms.

The discrete wavelet transform (DWT) remains a cornerstone for multiscale analysis of LRD traffic~\cite{Daubechies1992,Abry1998,Mallat2009}. Conventional orthonormal wavelets provide the mathematical safety net required for control systems -- offering perfect reconstruction (PR)\footnote{PR is the property that $\hat{x}[n] = x[n]$ with no aliasing or distortion.} and Parseval energy conservation. However, they rely on fixed, a priori designed filter banks (e.g., Haar, Daubechies) that cannot adapt to the evolving correlation structure of real network traffic, leading to suboptimal compression. Conversely, recent machine learning approaches introduce data-driven adaptability through neural or statistical representations~\cite{Boutaba2018,Szostak2021,Iraj2022,Aouedi2025DLTraffic}. Yet, these ``black-box'' methods often relax structural guarantees, introducing approximation errors that can lead to unpredictable behavior under perturbations or resource constraints.

This tension between data-driven adaptability and mathematical rigor defines the central challenge for application-aware telemetry: \emph{can we learn multiscale representations that remain orthonormal and provably stable while adapting to the statistics of network traffic?} Existing learned transforms generally fall into three categories: (i) unconstrained models that abandon orthogonality for flexibility, losing PR guarantees~\cite{wang2018multilevel, khan2018learning}; (ii) soft-constrained approaches that enforce properties via loss penalties, which hold only approximately and require laborious hyperparameter tuning~\cite{ha2021adaptive, wolter2021adaptive}; or (iii) structural methods that impose conjugate quadrature filter (CQF) constraints~\cite{michau2022fully} but do not guarantee exact orthogonality at intermediate training steps. Such approximations are insufficient for mission-critical DT: Parseval violations corrupt energy budgets, and imperfect reconstruction distorts the traffic's LRD signature, degrading the twin's predictive stability. Lezcano-Casado and Martínez-Rubio~\cite{lezcano2019cheap} do maintain exact orthogonality via exponential parametrization, but their framework targets recurrent neural networks (RNN) for temporal sequence modeling rather than multiscale signal decomposition. Although neural autoencoders achieve impressive rate-distortion on generic signals, they lack the interpretable multiscale structure and LRD-preservation guarantees that network state synchronization requires.

This work answers affirmatively by introducing a manifold-constrained optimization scheme where the framework enforces orthogonality at every
training iteration through polar projection onto the orthogonal manifold~\cite{sato2021,boumal2023,fei2025}, ensuring PR and Parseval energy conservation hold throughout learning. The mathematical foundation draws from the multiscale entanglement renormalization ansatz (MERA)~\cite{Vidal2007ER,Orus2014} -- a hierarchical tensor network (TN) from quantum many-body physics, reformulated here as a trainable cascade of local $2{\times}2$ orthogonal transformations. While the physics literature has treated wavelets as a mathematical analogy when synthesizing quantum states~\cite{Evenbly2016,Haegeman2018}, this work inverts the paradigm: by imposing constraints on MERA tensors and interpreting the resulting decomposition as a learnable filter bank, Theorem~\ref{thm:circuit_pu} (Section~\ref{sec:equivalence-PUFB}) establishes that MERA layers are mathematically equivalent to two-channel paraunitary filter banks at every decomposition level. This equivalence is not approximate or asymptotic -- it holds exactly, enabling a framework that unifies data-driven adaptability with the mathematical rigor necessary for reliable closed-loop operation. With polar projection ensuring orthogonality at every training iteration, the framework guarantees energy preservation and invertibility at all scales while retaining full adaptability to data statistics. When used for compression, it yields interpretable rate-distortion trade-offs by retaining a fraction $\rho$ of coefficients, offering empirical validation of its energy compaction properties.

Having motivated the need for adaptive wavelets with structural guarantees, the main contributions of this work are summarized as follows:
\begin{itemize}
	\item A formal equivalence between MERA TN and orthonormal paraunitary wavelet filter banks is established (Theorem~\ref{thm:circuit_pu}), bridging concepts from quantum many-body theory and multirate signal processing.
	
	\item A learning framework operating directly on the Stiefel manifold $\mathcal{O}(2)$ (the group of $2 \times 2$ orthogonal matrices) is introduced, enforcing PR and Parseval energy preservation via polar projection at every iteration. This eliminates the approximation errors inherent to soft-penalty methods~\cite{wang2018multilevel,wolter2021adaptive} and ensures orthogonality throughout training, unlike CQF-based approaches~\cite{michau2022fully} that allow intermediate coefficient drift or exponential parametrizations~\cite{lezcano2019cheap} designed for RNN stability.
	
	\item Experimental validation on six real-world backbone traffic traces spanning 2020--2025 (314~Mbps--1.75~Gbps) demonstrates 0.5--3.8~dB Peak Signal-to-Noise Ratio (PSNR) gains over fixed wavelet bases while preserving Hurst exponents within 95\% confidence intervals at 90\% compression, establishing superior LRD retention.
\end{itemize}

\paragraph*{Roadmap} 
The remainder of this paper is organized as follows:
\begin{itemize}
	\item \textbf{Section~\ref{sec:dtwins}}: DT telemetry compression -- requirements, bottleneck analysis, and scope.
	
	\item \textbf{Section~\ref{sec:background}}: Mathematical foundations.
	
	\item \textbf{Section~\ref{sec:mera-wavelet}}: MERA-inspired wavelet architecture.
	
	\item \textbf{Section~\ref{sec:equivalence-PUFB}}: Equivalence to paraunitary filter banks (Theorem \ref{thm:circuit_pu}).
	
	\item \textbf{Section~\ref{sec:learning-framework}}: Learning framework with manifold optimization.
	
	\item \textbf{Section~\ref{sec:experiments}}: Validation on real backbone traces.
	
	\item \textbf{Section~\ref{sec:conclusion}}: Conclusion and future directions.
\end{itemize}

\section{Digital Twin Synchronization: A Layered Perspective}
\label{sec:dtwins}

This section establishes the context and requirements that motivate the proposed framework. Fig.~\ref{fig:layered-architecture} illustrates the layered architecture, highlighting the telemetry compression layer addressed by this work.

\subsection{The Network Digital Twin Paradigm}
\label{sec:dt_paradigm}

Originally proposed in ~\cite{Grieves} as a digital representation to support the design and development of manufactured components, the concept of a Digital Twin (DT) has evolved into an essential component towards the next generation networks and services. The DT paradigm refers to the construction of continuously updated digital counterparts capable of mirroring the behavior and state of physical entities (or even purely virtual, or a physical-virtual hybrid entity). DT methodologies have been extended to the domain of communication infrastructures, emerging the Digital Twin Network (DTN) architectures ~\cite{zhou-nmrg-digitaltwin-network-concepts-00} and the Network Digital Twin (NDT) instances ~\cite{irtf-nmrg-network-digital-twin-arch-10}. An NDT constitutes a virtualized replica of an Original Network (ON), whether physical or virtual, and remains tightly coupled with it through  information exchange that enables near real-time state synchronism. Therefore, NDTs are able to support advanced functionalities such as online simulation, network design,  optimization, and AI-driven control and orchestration mechanisms, which can be exploited by the ON to enhance operational efficiency and overall performance. Moreover, the flexible and interoperable design of NDTs makes them suitable for deploying new network services. An NDT does not require dedicated physical equipment to be realized, but it can be instantiated either on specific computing resources or through virtualized resources and service infrastructures. Fixed hardware solutions can guarantee an exact replica of the original, at the same cost. In comparison, virtualization-based approaches enable significantly more dynamic control and allow the NDT to be tailored more easily to the requirements, and usually with a reduced cost.

Network Digital Twins represent an emerging paradigm for 6G network management, where a dynamic virtual replica of the physical infrastructure enables simulation-based optimization, capacity planning, and ``what-if'' analysis before deploying changes to production systems~\cite{Hesselbach2025DTN, Tariq2022, Kuruvatti2022, Wang2025}. The DT continuously ingests telemetry data from the physical network -- capturing traffic characteristics, queue states, and resource utilization -- to maintain synchronization between the virtual model and real-world dynamics~\cite{Tariq2022}.

Typically, a single Digital Twin is associated with an Original, say 1:1. However, multiple Digital Twins can also be instantiated in parallel (say 1:N), forming a DT farm in which each instance can be focused on analyzing different aspects of the system or evaluating alternative strategies in order to be compared. So, a DT farm enables the distribution of analytical tasks across several specialized DT instances. Each Digital Twin can be tailored to focus on a particular target. DT farms allow the analysis and comparison of alternative strategies under identical baseline conditions. Because all DTs originate from the same synchronized state of the original, their outcomes can be compared without affecting the source, reducing risks and accelerating the analysis results.

From a DT perspective, backbone telemetry constitutes the dominant synchronization bottleneck, as it aggregates traffic originating from radio, edge, and core domains into a unified stochastic process. As a result, distortions introduced at this layer propagate directly into the virtual model, affecting the fidelity of downstream simulation and optimization tasks.

The effectiveness of a DT hinges on \emph{synchronization fidelity}: the degree to which the virtual model accurately reflects the statistical and temporal properties of the physical network. This fidelity directly impacts the reliability of simulations used for critical decisions such as congestion control, routing optimization, and service-level agreement (SLA) enforcement.

High-fidelity synchronization requires continuous telemetry ingestion at temporal resolutions sufficient to capture traffic dynamics ranging from millisecond-scale microbursts to hour-scale session patterns. The backbone traces employed in this work (Section~\ref{sec:experiments}, Table~\ref{tab:mawi_traces_summary}) illustrate typical data volumes: at 1~ms sampling granularity, monitoring hundreds of concurrent links can generate gigabytes of raw telemetry per collection cycle\footnote{For example, 1~ms sampling over 15-minute windows yields ${\sim}9 \times 10^5$ samples per trace. At 64-bit precision, a single link generates ${\sim}$7.2~MB per interval; scaling to 200~links produces ${\sim}$1.4~GB per cycle~\cite{Kuruvatti2022}.}. While such overhead is negligible in overprovisioned core networks, it becomes relevant in bandwidth-constrained scenarios such as satellite backhaul or disaggregated RAN fronthaul~\cite{Wang2025}.

Beyond volume reduction, a more fundamental requirement is \emph{statistical fidelity}: compression schemes must preserve the invariants that govern network performance models -- most critically, the LRD structure of traffic (Section~\ref{sec:lrd_requirement}). This motivates the development of adaptive transforms that maintain structural guarantees while exploiting signal-specific correlations.

\subsection{Technical Requirements and Scope Delimitation}
\label{sec:scope}

This work addresses the \emph{telemetry compression layer} (Fig.~\ref{fig:layered-architecture}) within the broader DT synchronization pipeline. The contribution is a signal processing solution that operates on time-series telemetry streams (e.g., byte-rate samples at 1~ms granularity) and produces compressed representations suitable for transmission to DT infrastructure. The compression method must simultaneously achieve:

\begin{enumerate}
	\item \textbf{Rate-Distortion Efficiency:} Maximize reconstruction fidelity (PSNR) under bandwidth constraints.
	
	\item \textbf{Statistical Fidelity:} Preserve the Hurst exponent $H$ within estimator confidence intervals, ensuring that decompressed telemetry retains the LRD structure necessary for accurate queueing analysis.
	
	\item \textbf{Structural Guarantees:} Provide PR and Parseval energy conservation for predictable, deterministic behavior in mission-critical operations.
\end{enumerate}

\begin{figure}[t]
	\centering
	\begin{tikzpicture}[
		node distance=0.3cm,
		layer/.style={
			rectangle, 
			draw=black, 
			minimum width=0.75\columnwidth,  
			minimum height=1.1cm,
			text centered,
			font=\small,
			text width=0.85\columnwidth,
			align=center
		},
		highlight/.style={
			layer,
			fill=blue!15,
			draw=blue!60,
			line width=1.5pt
		}
		]
		
		\node[layer, fill=gray!10] (app) {
			\textbf{Digital Twin Application Layer}\\[1pt]
			{\footnotesize Simulation, Optimization, What-if Analysis}
		};
		
		\node[layer, fill=gray!10, below=of app] (sync) {
			\textbf{DT Synchronization Layer}\\[1pt]
			{\footnotesize State Updates, Model Calibration}
		};
		
		\node[highlight, below=of sync] (compress) {
			\textbf{Telemetry Compression Layer} {\small\textit{(This Work)}}\\[1pt]
			{\footnotesize MERA-Wavelet Codec: Adaptive, LRD-Preserving}
		};
		
		\node[layer, fill=gray!10, below=of compress] (collect) {
			\textbf{Data Collection Layer}\\[1pt]
			{\footnotesize In-band Network Telemetry, Streaming Telemetry}
		};
		
		\node[layer, fill=gray!10, below=of collect] (physical) {
			\textbf{Physical Network Layer}\\[1pt]
			{\footnotesize Routers, Switches, Links, Queues}
		};
		
		\draw[->, >=stealth, thick] (physical.north) -- (collect.south);
		\draw[->, >=stealth, thick] (collect.north) -- (compress.south);
		\draw[->, >=stealth, thick, blue!60, line width=1.5pt] (compress.north) -- (sync.south);
		\draw[->, >=stealth, thick] (sync.north) -- (app.south);
		
		\draw[->, >=stealth, thick, dashed] (app.west) -- ++(-0.7,0) |- (physical.west);
		
	\end{tikzpicture}
	\caption{Layered architecture for NDT synchronization. The telemetry compression layer (highlighted) provides the interface between raw network measurements and the virtual model. This work contributes the adaptive MERA-wavelet codec operating at this layer. }
	\label{fig:layered-architecture}
\end{figure}
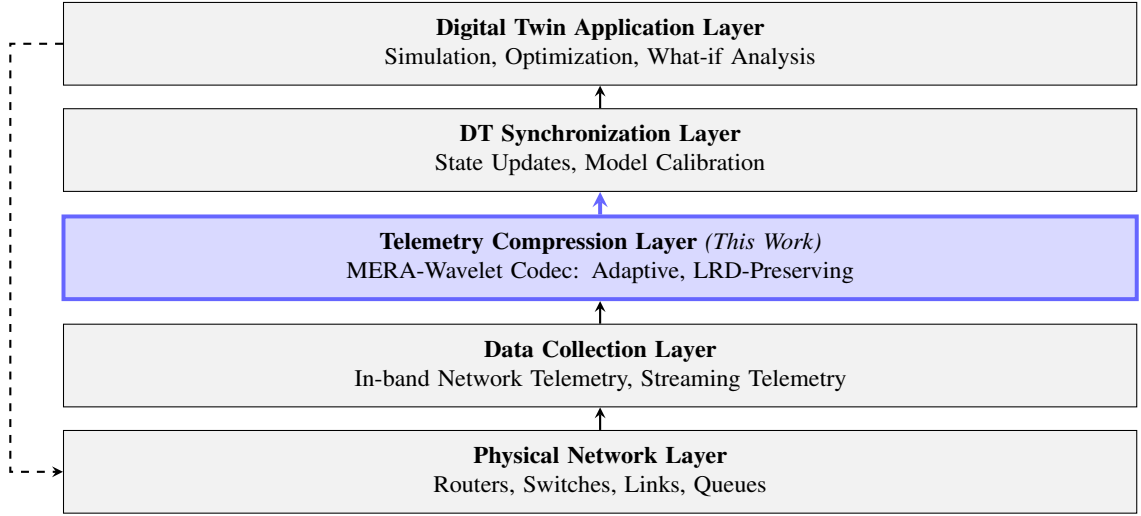

\paragraph{Standard Wavelets vs.\ Learned Approaches} Classical orthogonal wavelets (Haar, Daubechies, Coiflets) satisfy requirement~(3) through their paraunitary properties but achieve suboptimal performance on requirements~(1) and~(2) due to fixed filter designs optimized for polynomial smoothness rather than power-law correlations. Conversely, fully learned neural approaches (e.g., autoencoders) may excel at~(1) but sacrifice~(3) by relaxing orthogonality constraints, introducing approximation errors incompatible with safety-critical DT applications.

The proposed MERA-inspired adaptive wavelets reconcile this trade-off by learning filter banks from data while maintaining paraunitary guarantees through manifold-constrained optimization (Section~\ref{sec:learning-framework}). By adapting to the specific spectral characteristics of backbone traffic, the method achieves superior energy compaction (up to 3.8~dB PSNR gain over fixed wavelets, Section~\ref{sec:experiments}) while preserving $H$ within 95\% confidence intervals at 90\% compression (Section~\ref{sec:statistical_fidelity}).

\paragraph{Architectural Positioning} The codec is DT-agnostic: it interfaces with any simulator or emulator that ingests traffic time series as input -- including discrete-event simulators (e.g., NS-3), queueing models, fluid-flow approximations, or hardware-in-the-loop testbeds (e.g., MININET). Rather than prescribing a particular DT architecture, it provides a reusable compression module that preserves the statistical properties required across diverse modeling frameworks.

\paragraph{Validation Strategy} Following standard practice in source coding research -- where video codecs are validated using rate-distortion metrics on benchmark datasets without implementing full streaming protocol stacks -- the proposed codec is validated using PSNR (rate-distortion efficiency) and Hurst exponent deviation $|\Delta H|$ (statistical fidelity) on six years of MAWI backbone traces~\cite{mawi_archive} spanning 314~Mbps--1.75~Gbps with $H \in [0.77, 0.93]$. Baselines include classical orthogonal wavelets: Haar, Daubechies-4, Coiflet-3, Symmlet-8, and Biorthogonal-4.4.

This validation demonstrates that compressed telemetry retains the statistical properties necessary for downstream DT models, independent of specific simulator implementations. Integration with full DT frameworks -- including topology modeling, routing protocols, and closed-loop control -- represents important future work (Section~\ref{sec:conclusion}) but falls outside the scope of this signal processing contribution.

\section{Mathematical Background}
\label{sec:background}

This section establishes the mathematical foundations underlying the proposed framework: orthogonal transformations that preserve signal energy, the MERA TN architecture that organizes these transformations hierarchically, and the Stiefel manifold on which constrained optimization is performed to maintain orthogonality throughout learning.

\subsection{Unitary and Orthogonal Transformations}
\label{subsec:unitary}

\begin{definition}[Unitary Transformation]\label{def:unitary}
	A linear operator $U: \mathcal{H} \to \mathcal{H}$ on a complex inner product space $\mathcal{H}$ is \emph{unitary} if it preserves inner products:
	$\langle Ux, Uy \rangle = \langle x, y \rangle$ for all $x, y \in \mathcal{H}$,
	equivalently $U^\dagger U = I$, where ${}^\dagger$ denotes the conjugate transpose (Hermitian adjoint).
\end{definition}

\begin{definition}[Orthogonal Transformation]\label{def:orthogonal}
	For real-valued spaces $\mathcal{H} = \mathbb{R}^n$, a matrix $U \in \mathbb{R}^{n \times n}$ is orthogonal if $U^\dagger U = I$. The set of all such matrices forms the orthogonal group $\mathcal{O}(n) = \{U \in \mathbb{R}^{n \times n} \mid U^\dagger U = I\}$.
\end{definition}

\begin{remark}[Notational Convention]\label{rem:dagger-notation}
	The dagger symbol $A^\dagger$ denotes the conjugate transpose, following conventions in quantum TN~\cite{Vidal2008,Evenbly2016}. For real matrices, $A^\dagger = A^\mathsf{T}$. This notation is retained throughout to emphasize the structural connection to MERA formalism, while acknowledging that $U^\dagger = U^\mathsf{T}$ in the real-valued implementation ($U_\ell \in \mathcal{O}(2)$).
\end{remark}

Orthogonality ensures three critical properties for adaptive wavelets: (i)~energy conservation via the Parseval identity, (ii)~PR through $U^\dagger$, and (iii)~numerical stability under composition ($\|U\| = 1$). These guarantees are maintained throughout optimization via polar projection onto $\mathcal{O}(2)$ (Section~\ref{sec:learning-framework}), distinguishing the proposed framework from approaches where orthogonality is imposed only approximately~\cite{wang2018multilevel,michau2022fully}.

\subsection{MERA Tensor Networks}
\label{subsec:mera}

MERA TN were introduced by Vidal~\cite{Vidal2007ER} to efficiently represent quantum systems exhibiting scale-invariant correlations with power-law decay -- a property that directly parallels LRD in network traffic. MERA organizes computation into hierarchical layers, each applying:
\begin{enumerate}
	\item \textbf{Disentanglers:} local unitary transformations removing short-range correlations before coarse-graining.
	\item \textbf{Isometries:} linear maps satisfying $U^\dagger U = I$ that reduce degrees of freedom (typically by factor two) while preserving large-scale structure.
\end{enumerate}

This alternating disentangle--coarsen procedure across $L$ layers directly parallels dyadic wavelet decomposition: each MERA layer corresponds to a resolution level, with isometries playing the role of analysis filters. As shown by Reyes and Stoudenmire~\cite{Reyes2021}, MERA can learn hierarchical correlations across resolutions, bridging quantum renormalization with deep-learning principles. Section~\ref{sec:equivalence-PUFB} formalizes the equivalence between MERA layers and paraunitary filter banks (Theorem~\ref{thm:circuit_pu}), enabling adaptive wavelets with exact PR and energy conservation guarantees.

\subsection{The Stiefel Manifold}
\label{subsec:stiefel_background}

The orthogonality requirements of paraunitary filter banks frame learning as a constrained optimization problem on smooth manifolds. Specifically, the Stiefel manifold $\mathrm{St}(n,k)$ is
defined as the set of matrices with orthonormal columns:
\begin{equation}
	\mathrm{St}(n, k) = \{ U \in \mathbb{C}^{n \times k} \mid U^\dagger U = I_k \}.
\end{equation}

Standard optimizers (SGD, Adam) compute updates in ambient Euclidean space.
A linear update
\begin{equation}
	\widetilde{U}_{t+1} = U_t - \eta \nabla \mathcal{L}(U_t)
\end{equation}
generally violates orthogonality, since the gradient may contain a nonzero
component normal to the manifold. To maintain structural integrity, the
Euclidean step is followed by a polar retraction~\cite{absil2008}:
\begin{equation}
	U_{t+1}
	=
	\mathcal{R}(\widetilde{U}_{t+1})
	=
	\widetilde{U}_{t+1}
	\left( \widetilde{U}_{t+1}^\dagger \widetilde{U}_{t+1} \right)^{-1/2},
\end{equation}
which projects onto $\mathrm{St}(n,k)$, ensuring exact orthogonality at every
iteration (Section~\ref{sec:learning-framework}).

It is emphasized that the architecture considered in this work corresponds to a disentangler-free, tree-structured MERA-inspired network, rather than a full MERA in the strict tensor-network-theoretic sense. This restriction is deliberate, as it preserves the exact paraunitary structure required for perfect reconstruction.

\section{Mera-Inspired Wavelet Architecture}
\label{sec:mera-wavelet}

This section introduces a unified framework for adaptive orthonormal wavelets tailored to LRD network traffic. The key idea is to reinterpret the MERA architecture as structured parameterizations of paraunitary filter banks, thereby enabling learnable multiscale representations that remain orthonormal by design. Section \ref{subsec:need_adaptive} motivates the need for adaptive multiscale models in the presence of LRD traffic, while Section~\ref{subsec:mera_layers} introduces the MERA-inspired orthogonal layers (Definition~\ref{def:mera_layer}) that provide the architectural foundation.

\subsection{Adaptive Multiscale Models for LRD Traffic}
\label{subsec:need_adaptive}

An accurate characterization of the Hurst exponent $H$ from traffic measurements is critical for capacity planning, queueing analysis, and DT synchronization in 6G systems. The DWT provides a natural framework for both analyzing and representing LRD signals through multiresolution decomposition, whose hierarchical structure directly mirrors the scale-invariant correlation patterns characteristic of fractal network traffic. Fig.~\ref{fig:mra} illustrates the Mallat pyramid: at each scale $\ell$ ($\ell=1,2,3$), the signal is recursively split into approximation coefficients $a_{\ell}$ (low-pass filter $g$) and detail coefficients $d_{\ell}$ (high-pass filter $h$), with downsampling by two ($\downarrow 2$) at each stage. This dyadic decomposition not only aligns naturally with the self-similar structure of LRD processes but also enables robust Hurst exponent estimation via wavelet variance scaling~\cite{Abry1998}, making wavelets the standard tool for LRD traffic analysis and compression.

\begin{figure}[!ht]
	\centering
	\includegraphics[width=0.6\columnwidth]{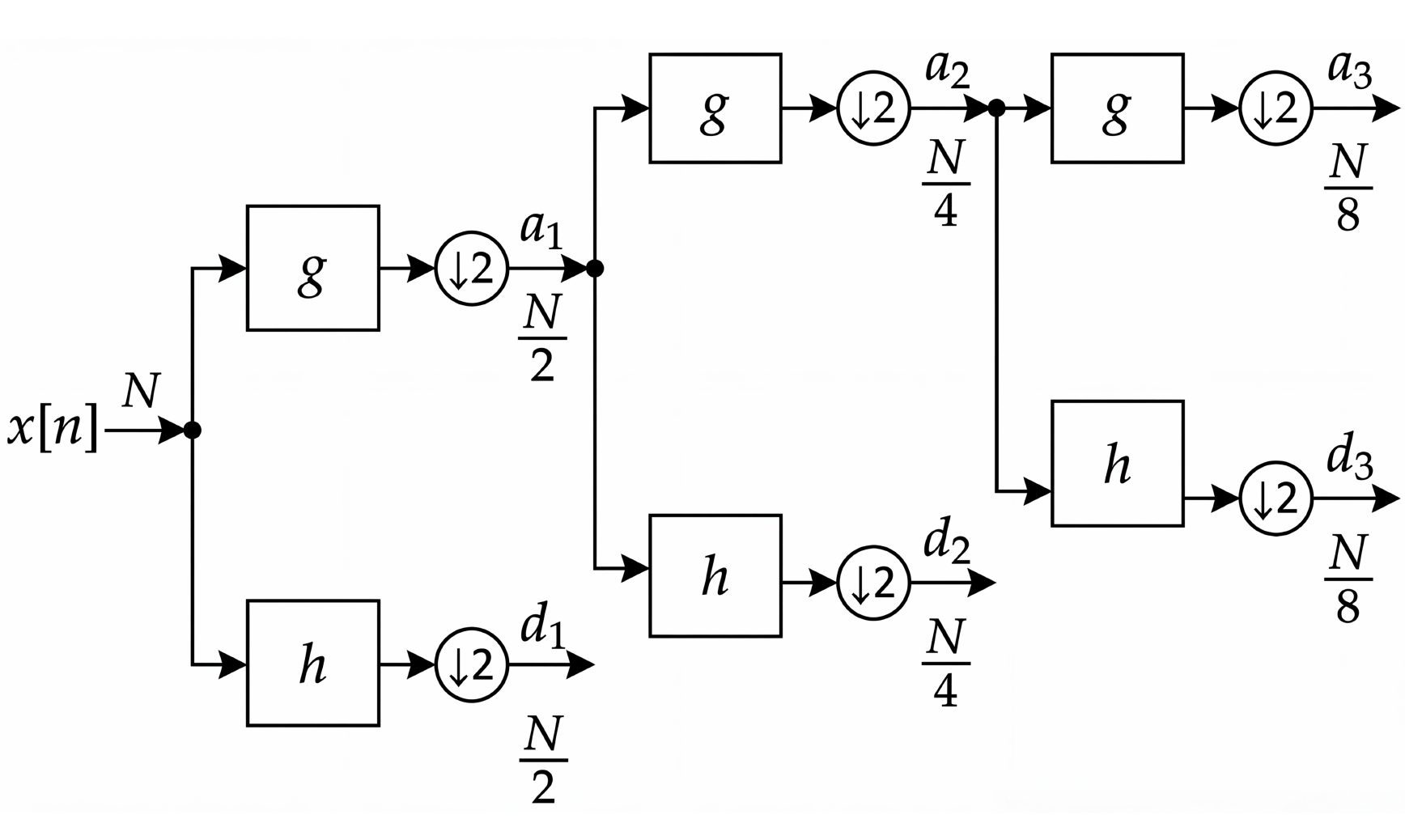} 
	\caption{Multiresolution analysis (MRA) illustrating recursive approximation/detail splitting with decimation by two.
		The input discrete-time signal $x_n$ is successively filtered by the low-pass filter $g$ (scaling function) and high-pass filter $h$ (wavelet function), followed by downsampling by a factor of two. The approximation stream propagates upward through all levels, while the detail streams are extracted at each corresponding scale.
		At each stage, the signal length is halved ($N \rightarrow N/2 \rightarrow N/4 \rightarrow N/8$), forming the dyadic tree structure characteristic of the DWT \cite{Mallat2009}.}
	\label{fig:mra}
\end{figure}

Despite this structural alignment, DWT relies on fixed, pre-designed filter banks (Haar, Daubechies, Coiflets, Symmlets) optimized for generic smoothness assumptions\footnote{Mallat~\cite{Mallat2009} shows that classical wavelets achieve optimal approximation rates for functions in Besov spaces -- those with 
bounded derivatives admitting local polynomial approximations. Network traffic violates this assumption due to its impulsive, fractal structure.} \cite{Mallat2009}. While these bases provide rigorous guarantees, their filters are hand-crafted to maximize vanishing moments, not to capture the power-law correlations $\phi(k) \sim k^{-\beta}$ characteristic of LRD traffic. Consequently, fixed wavelets achieve suboptimal energy compaction on backbone traces: detail coefficients retain significant energy that could be concentrated into approximations, degrading rate-distortion performance under bandwidth constraints. The central challenge is thus to \emph{learn} wavelet filters adapted to traffic-specific correlation structures while preserving the mathematical guarantees that make wavelets reliable for mission-critical telemetry.

\subsection{MERA-Inspired Orthogonal Layers}
\label{subsec:mera_layers}

A hierarchical architecture inspired by MERA TN is introduced to provide a structured parameterization of orthonormal wavelets. The design is guided by three principles: (1)~\textbf{Locality} -- transformations operate on disjoint pairs; (2)~\textbf{Orthogonality} -- all operations preserve inner products; and (3)~\textbf{Hierarchy} -- layers correspond to dyadic scales.

\begin{definition}[MERA-Inspired Orthogonal Layer]
	\label{def:mera_layer}
	Let $\mathbf{x} \in \mathbb{R}^N$ be a discrete-time signal with $N = 2^L$, where $N$ is the number of samples and $L$ is the maximum decomposition level. A MERA layer at scale $\ell$ applies a $2 \times 2$ orthogonal matrix $U_\ell \in \mathcal{O}(2)$ to disjoint pairs of samples  (implicit downsampling by two ($\downarrow 2$)):
	\begin{equation}
		\label{eq:mera_layer}
		\begin{bmatrix} a_k^{(\ell)} \\ d_k^{(\ell)} \end{bmatrix}
		= U_\ell 
		\begin{bmatrix} x_{2k} \\ x_{2k+1} \end{bmatrix},
		\qquad
		k = 0,\ldots,\frac{N}{2^{\ell }}-1,
	\end{equation}
	where the outputs $\mathbf{a}=\{a_k^{(\ell)}\}$ and $\mathbf{d}=\{d_k^{(\ell)}\}$ have length $N/2^{\ell}$ each and are the decimated approximation and detail coefficients, respectively.
\end{definition}

\begin{example}[MERA Layer Computation]
	\label{ex:mera-layer-computation}
	Consider input $\mathbf{x} = [1, 2, 3, 4]^\top$ and orthogonal matrix $U_1 = \frac{1}{\sqrt{2}}\begin{bmatrix} 1 & 1 \\ 1 & -1 \end{bmatrix}$ (Haar). Applying Eq.~\eqref{eq:mera_layer}:
	\begin{align*}
		\begin{bmatrix} a_0^{(1)} \\ d_0^{(1)} \end{bmatrix} &= U_1 \begin{bmatrix} x_0 \\ x_1 \end{bmatrix} = \frac{1}{\sqrt{2}} \begin{bmatrix} 1 & 1 \\ 1 & -1 \end{bmatrix} \begin{bmatrix} 1 \\ 2 \end{bmatrix} = \frac{1}{\sqrt{2}} \begin{bmatrix} 3 \\ -1 \end{bmatrix} \\
		\begin{bmatrix} a_1^{(1)} \\ d_1^{(1)} \end{bmatrix} &= U_1 \begin{bmatrix} x_2 \\ x_3 \end{bmatrix} = \frac{1}{\sqrt{2}} \begin{bmatrix} 1 & 1 \\ 1 & -1 \end{bmatrix} \begin{bmatrix} 3 \\ 4 \end{bmatrix} = \frac{1}{\sqrt{2}} \begin{bmatrix} 7 \\ -1 \end{bmatrix}
	\end{align*}
	Yielding approximation $\mathbf{a} = \frac{1}{\sqrt{2}}[3, 7]^\top$ and detail $\mathbf{d} = \frac{1}{\sqrt{2}}[-1, -1]^\top$. Energy conservation: $\|\mathbf{x}\|^2 = 30 = \|\mathbf{a}\|^2 + \|\mathbf{d}\|^2 = 29 + 1 = 30$ \checkmark
\end{example}

Orthogonality $U_\ell^\dagger U_\ell = I$ (Definition~\ref{def:orthogonal}) ensures (i) local pairwise energy conservation $a_k^2 + d_k^2 = x_{2k}^2 + x_{2k+1}^2$; (ii) 
Parseval identity $\|\mathbf{x}\|^2 = \|\mathbf{a}^{(L)}\|^2  + \sum_{\ell=1}^L\|\mathbf{d}^{(\ell)}\|^2$; and (iii) PR via $U_\ell^\dagger$. A complete $L$-level ($L=4$) cascade (Fig.~\ref{fig:mera-circuit}) applies these layers recursively, producing $\{\mathbf{a}^{(L)}, \mathbf{d}^{(L)}, \ldots, \mathbf{d}^{(1)}\}$. The resulting analysis operator $\mathcal{A}\in\mathcal{O}(N)$ inherits all guarantees with $\mathcal{O}(N)$ complexity via decimation, matching fast wavelet transforms.

\begin{figure*}[!t]
	\centering
	\includegraphics[width=0.95\textwidth]{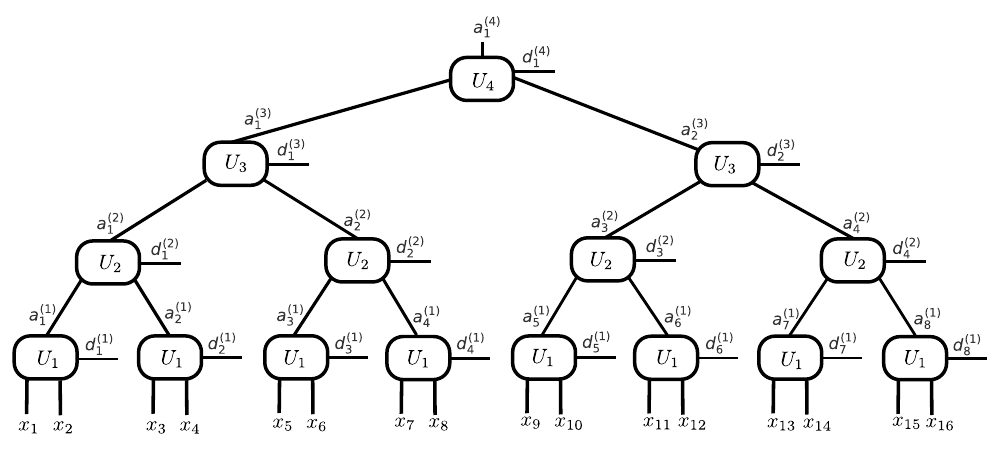}
	\caption{MERA-inspired wavelet circuit with four dyadic levels ($L=4$). 
		The input signal samples $x_1,\dots,x_{16}$ feed the first layer, which consists of parallel $2\times2$ orthogonal blocks $U_1$ acting on disjoint pairs. 
		At each level, the approximation outputs $a^{(\ell)}$ propagate upward through the hierarchy, while the detail outputs $d^{(\ell)}$ are extracted at their respective scales. 
		This hierarchical structure parallels the DWT (Fig.~\ref{fig:mra}) but employs \emph{learnable} transformation blocks $U_\ell$, $\ell=1,2,3,4$.}
	\label{fig:mera-circuit}
\end{figure*}

\paragraph*{Summary}
This section established the MERA-inspired architecture for adaptive wavelets. Section \ref{subsec:need_adaptive} motivated the need for adaptive multiscale models tailored to LRD traffic. Section \ref{subsec:mera_layers} introduced the orthogonal layer structure (Definition \ref{def:mera_layer}) that provides hierarchical decomposition while preserving energy conservation. Section V establishes the exact mathematical equivalence between these layers and classical paraunitary filter banks.

\section{Equivalence to Paraunitary Filter Banks}
\label{sec:equivalence-PUFB}

This section establishes the exact equivalence between MERA-inspired layers and two-channel paraunitary wavelet filter banks. Section \ref{subsec:poly-theory} introduces the polyphase theory background necessary for this equivalence. Section \ref{subsec:main-equivalence} presents the main result (Theorem \ref{thm:circuit_pu}) demonstrating that MERA layers with constant orthogonal matrices are mathematically equivalent to two-tap paraunitary filter banks. Section \ref{subsec:variational} formulates the manifold-constrained learning objective. Together, these results provide the theoretical foundation upon which the variational learning and optimization procedures developed in subsequent sections are built.

\subsection{Polyphase Theory Background}
\label{subsec:poly-theory}

\begin{definition}[Paraunitary Filter Bank]
	\label{def:PUFB}
	A two-channel filter bank with polyphase matrix
	\begin{equation}
		\label{eq:polyphase-matrix}
		\mathbf{E}(z) =
		\begin{bmatrix}
			G_0(z) & G_1(z) \\
			H_0(z) & H_1(z)
		\end{bmatrix}
	\end{equation}
	is said to be \emph{paraunitary} if
	\begin{equation}
		\label{eq:paraunitary}
		\mathbf{E}(z)\,\mathbf{E}^\dagger(z^{-1}) = \mathbf{I}.
	\end{equation}
	This condition ensures:
	\begin{enumerate}
		\item \textbf{PR:} synthesis filters
		$\tilde{G}(z) \triangleq G(z^{-1})$ and $\tilde{H}(z) \triangleq H(z^{-1})$
		satisfy $\hat{x}(z) = x(z)$;
		\item \textbf{Parseval energy conservation:}
		$\|\mathbf{x}\|^2 = \|\mathbf{a}^{(L)}\|^2  + \sum_{\ell=1}^L\|\mathbf{d}^{(\ell)}\|^2$;
		\item \textbf{Frequency-domain power complementarity:}
		$|G(\omega)|^2 + |H(\omega)|^2 = 2$.
	\end{enumerate}
\end{definition}

In multirate signal processing (Fig.~\ref{fig:qmf_diagrams}), PR requires the polyphase matrix~\eqref{eq:polyphase-matrix} to satisfy the paraunitary condition~\eqref{eq:paraunitary}~\cite{vaidyanathan1993}.

\textbf{Intuition:} A MERA-inspired layer applies the same $2 \times 2$ orthogonal matrix $U_\ell$ to disjoint pairs of samples -- that is, it \emph{intrinsically operates in the polyphase domain}, where filtering and decimation collapse into a single matrix multiplication. This pairwise block transform is equivalent to a two-channel paraunitary filter bank with two-tap finite impulse response (FIR) analysis filters. The special case where $U_\ell$ exhibits quadrature mirror filter (QMF) structure, combined with maximum DC gain, uniquely yields the Haar wavelet\footnote{Strictly speaking, the classical QMF structure ($H(z) = G(-z)$) cannot simultaneously achieve PR and linear phase with FIR filters, except for the trivial Haar case, as proven by Vaidyanathan \cite{vaidyanathan1993}. The filters employed in this work belong to the class of CQF introduced by Smith and Barnwell \cite{SmithBarnwell1986} -- also referred to as \textit{paraunitary QMF banks} by Vaidyanathan. }.


\begin{figure}[!t]
	\centering
	\includegraphics[width=0.8\columnwidth]{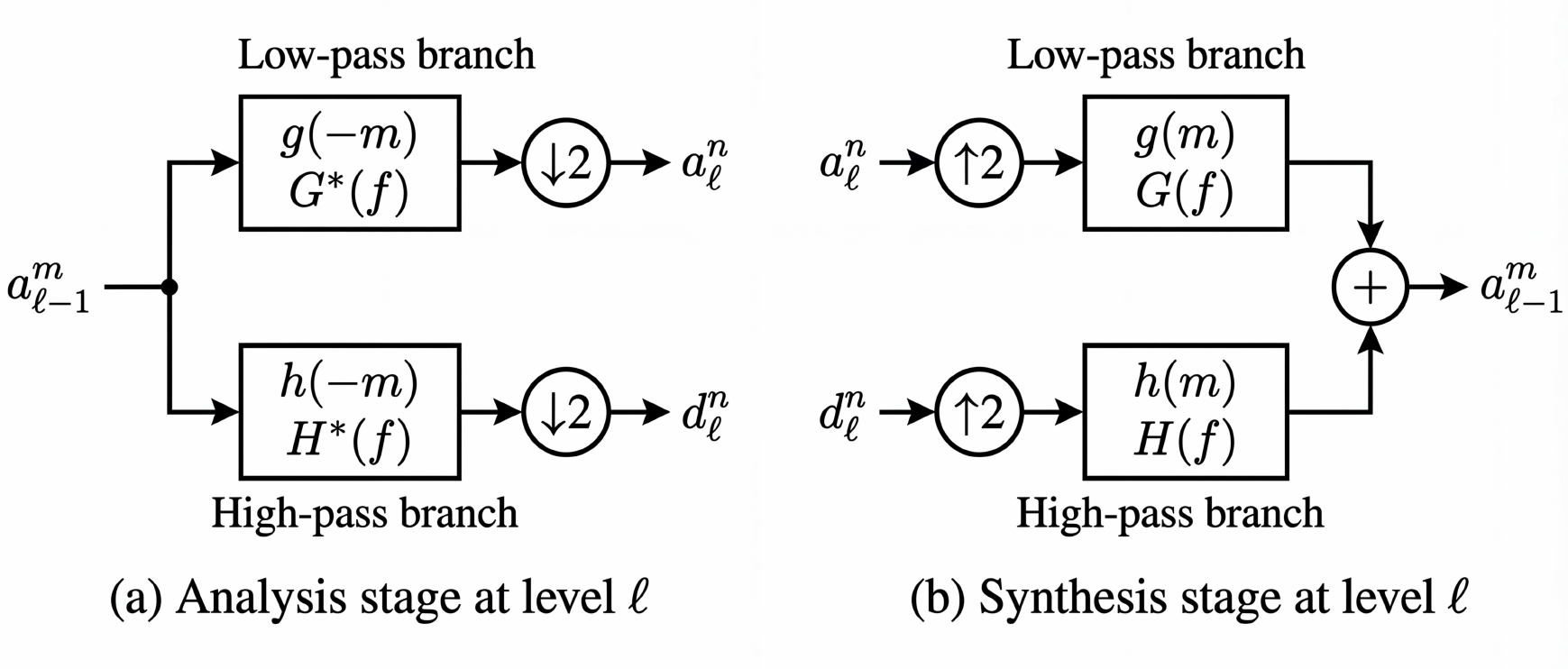}
	\caption{Two-channel CQF system. (a): Analysis stage splits input $a_{\ell-1}(m)$ via filters $g(-m)$, $h(-m)$ and downsampling by a factor of two ($\downarrow 2$), producing $a_\ell(n)$ and $d_\ell(n)$. (b): Synthesis stage upsamples by a factor of two ($\uparrow 2$), filters via $g(n)$, $h(n)$, and sums to reconstruct $a_{\ell-1}(m)$. 
		PR is characterized in the polyphase domain by a paraunitary matrix, a condition that will be guaranteed by orthogonal $U_\ell$ (Theorem~\ref{thm:circuit_pu}).}
	\label{fig:qmf_diagrams}
\end{figure}

\paragraph{Polyphase Decomposition.}
In multirate filter bank theory, the type-1 polyphase decomposition represents a filter $G(z)=\sum_n g[n]z^{-n}$ by separating its even- and odd-indexed coefficients. Following the notation of Vaidyanathan~\cite{vaidyanathan1993}, $G(z)$ denotes the full analysis filter, while $G_0(z)$ and $G_1(z)$ denote its even and odd polyphase components, respectively, defined as
\begin{equation}
	\label{eq:polyphase_def}
	G(z) = G_0(z^2) + z^{-1}G_1(z^2),
\end{equation}
with $G_0(z) = \sum_k g[2k]z^{-k}$ and $G_1(z) = \sum_k g[2k+1]z^{-k}$.

For a $2 \times 2$ polyphase matrix, the entries admit two equivalent representations. The first expresses the matrix directly in terms of the polyphase components of the analysis filters,
\begin{equation}
	\mathbf{E}(z) =
	\begin{bmatrix}
		G_0(z) & G_1(z) \\[2pt]
		H_0(z) & H_1(z)
	\end{bmatrix},
	\label{eq:polyphase_filters}
\end{equation}
whereas an alternative, more structural representation employs generic polyphase entries,
\begin{equation}
	\mathbf{E}(z) =
	\begin{bmatrix}
		E_{00}(z) & E_{01}(z) \\[2pt]
		E_{10}(z) & E_{11}(z)
	\end{bmatrix}.
	\label{eq:polyphase_entries}
\end{equation}
Although \eqref{eq:polyphase_filters} and \eqref{eq:polyphase_entries} are algebraically equivalent, the generic notation in \eqref{eq:polyphase_entries} emphasizes the polyphase matrix as the primary structural object.

The correspondence between the two notations is given by $E_{00}(z) \equiv G_0(z)$, $E_{01}(z) \equiv G_1(z)$, $E_{10}(z) \equiv H_0(z)$, and $E_{11}(z) \equiv H_1(z)$.
Applying the polyphase decomposition~\eqref{eq:polyphase_def} to both analysis filters yields:
\begin{align}
	G(z) &= E_{00}(z^2) + z^{-1}E_{01}(z^2)
	= G_0(z^2) + z^{-1}G_1(z^2), \label{eq:poly1} \\
	H(z) &= E_{10}(z^2) + z^{-1}E_{11}(z^2)
	= H_0(z^2) + z^{-1}H_1(z^2). \label{eq:poly2}
\end{align}

The usefulness of this representation follows from the Noble identities~\cite{vaidyanathan1993}, which establish that filtering by $H(z^2)$ followed by decimation by 2 ($\downarrow 2$) is equivalent to decimation by 2 ($\downarrow 2$) followed by filtering by $H(z)$. This commutation property allows the analysis outputs $A(z)$ and $D(z)$ to be computed directly in the decimated (polyphase) domain:
\begin{equation}
	\begin{bmatrix} A(z) \\ D(z) \end{bmatrix}
	=
	\mathbf{E}(z)
	\begin{bmatrix} X_0(z) \\ X_1(z) \end{bmatrix},
\end{equation}
where $X_0(z)$ and $X_1(z)$ denote the even and odd polyphase components of the input $X(z)=X_0(z^2)+z^{-1}X_1(z^2)$.

When $\mathbf{E}(z)\equiv U$ is \emph{constant} (i.e., $z$-independent), so that all polyphase entries satisfy $E_{ij}(z)\equiv u_{ij}$, equations \eqref{eq:poly1}--\eqref{eq:poly2} reduce to length-2 FIR analysis filters,
\begin{equation}
	G(z) = u_{00} + u_{01}z^{-1},
	\quad
	H(z) = u_{10} + u_{11}z^{-1}.
	\label{eq:two_tap_filters}
\end{equation}
This constant-polyphase structure is precisely the one induced by MERA layers, as demonstrated next.

\subsection{Main Equivalence Result}
\label{subsec:main-equivalence}

\begin{theorem}[Architectural Equivalence]
	\label{thm:circuit_pu}
	A MERA-inspired layer (Definition~\ref{def:mera_layer}) is equivalent to a two-channel paraunitary filter bank whose polyphase representation is a constant orthonormal matrix $\mathbf{E}(z) \equiv U_\ell$:
	\begin{equation}
		\begin{bmatrix} A(z) \\ D(z) \end{bmatrix} = \begin{bmatrix} g_0 & g_1 \\ h_0 & h_1 \end{bmatrix} \begin{bmatrix} X_0(z) \\ X_1(z) \end{bmatrix}
		\label{eq:polyphase_relation}
	\end{equation}
	where $z=e^{jw}$, $X_0(z) = \sum_k x_{2k}z^{-k}$, $X_1(z) = \sum_k x_{2k+1}z^{-k}$ are even/odd polyphase components, and $U_\ell = \begin{bsmallmatrix} g_0 & g_1 \\ h_0 & h_1 \end{bsmallmatrix} \in \mathcal{O}(2)$.
\end{theorem}

\textit{The proof is given in Appendix~\ref{appendix:proof_equivalence}.}

\begin{corollary}[PR QMF as a special case]
	\label{cor:qmf}
	The PR QMF constraint $h[n] = (-1)^n g[N-1-n]$ ($N=2$ for a two-tap FIR filter)~\cite{Mallat2009, vaidyanathan1993} arises as a special case of the framework when additional symmetry is imposed. This condition ensures that the highpass filter $H(z)$ is derived from the lowpass $G(z)$ through frequency reversal, thereby guaranteeing both orthogonality and alias cancellation.
	
	Under the assumptions of Theorem~\ref{thm:circuit_pu}, imposing the quadrature-mirror symmetry with $N=2$ gives
	\begin{equation}
		h[n] = (-1)^n g[1-n] 
		\quad \Longleftrightarrow \quad 
		H(z) = -z^{-1}G(-z^{-1}).
		\label{eq:qmf_constraint}
	\end{equation}
	For two-tap FIR analysis filters $G(z) = g_0 + g_1 z^{-1}$ and $H(z) = h_0 + h_1 z^{-1}$, the QMF relation yields
	\begin{equation}
		h_0 = g_1, \qquad h_1 = -g_0,
		\label{eq:qmf_coeffs}
	\end{equation}
	so that the polyphase matrix takes the form
	\begin{equation}
		U_\ell = 
		\begin{bmatrix} 
			g_0 & g_1 \\ 
			g_1 & -g_0 
		\end{bmatrix}.
		\label{eq:qmf_polyphase}
	\end{equation}
	
	Furthermore, the Haar wavelet is the unique real two-tap FIR filter bank satisfying both PR, QMF paraunitarity, and maximal DC gain (equivalently, $g_0 = g_1$ (see Appendix~\ref{appendix:proof_qmf} for proof)).	
\end{corollary}

\subsection{Manifold-Constrained Learning Objective}
\label{subsec:variational}

This subsection formulates the learning problem associated with the MERA-inspired paraunitary architecture. The optimization objective is to learn scale-dependent orthogonal transformations that maximize energy compaction in LRD traffic while guaranteeing PR and Parseval energy conservation.

Let $\theta = \{U_\ell\}_{\ell=1}^{L}$ denote the learnable parameters, where each $U_\ell \in \mathcal{O}(2)$. For an input signal $\mathbf{x}$, the analysis transform $\mathcal{A}_\theta$ produces the coefficient set $\{\mathbf{a}^{(L)}, \mathbf{d}^{(1)}, \ldots, \mathbf{d}^{(L)}\}$. Signal reconstruction is given by $\hat{\mathbf{x}} = \mathcal{S}_\theta(\mathcal{A}_\theta(\mathbf{x}))$, where $\mathcal{S}_\theta = \mathcal{A}_\theta^\dagger$.

The loss function $\mathcal{L}(\{U_\ell\})$ promotes sparse multiscale representations by concentrating signal energy into a small number of approximation coefficients while penalizing the aggregate magnitude of detail coefficients across all scales:
\begin{equation}
	\label{eq:loss_variational}
	\min_{\{U_\ell \in \mathcal{O}(2)\}} \mathcal{L} = \underbrace{\lambda_{\text{sparse}} \sum_{\ell=1}^{L} \frac{1}{N_\ell} \|\mathbf{d}^{(\ell)}\|_1}_{\text{sparsity term}} + \underbrace{\frac{\lambda_{\text{MSE}}}{N} \|\mathbf{x} - \hat{\mathbf{x}}\|_2^2}_{\text{reconstruction term (MSE)}}
\end{equation}
where $N_\ell \triangleq \text{card}(\mathbf{d}^{(\ell)}) = N/2^\ell$ denotes the number of detail coefficients at scale $\ell$.

\paragraph{Loss function terms} The \emph{sparsity term} ($\ell_1$ norm) promotes energy compaction into few large-magnitude coefficients. 
The \emph{reconstruction term} (MSE) penalizes mismatch between the input signal and its reconstruction from the full coefficient set. This term is optional: when $\lambda_{\text{MSE}} = 0$, training minimizes only the sparsity objective $\sum_\ell \|\mathbf{d}^{(\ell)}\|_1$. Since orthogonality is enforced by projection after each update, PR is guaranteed regardless of whether the MSE term is active.

\paragraph{Normalization} The sparsity term is normalized by $N_\ell$ to ensure that $\lambda_{\text{sparse}}$ remains invariant with respect to changes in the decomposition depth $L$ and the window size $N$. Without this normalization, the effective weight of the sparsity penalty would decrease as decimation reduces the number of coefficients at coarser scales, requiring retuning of $\lambda_{\text{sparse}}$ for different configurations.

\paragraph{Manifold Constraint}

The constraint $U_\ell \in \mathcal{O}(2)$ defines a smooth Riemannian manifold. This formulation differs from unconstrained empirical risk minimization, as orthogonality is enforced \emph{structurally} via polar projection onto the orthogonal group (Section~\ref{sec:learning-framework}) rather than through soft penalty terms. As a result, perfect reconstruction and Parseval energy conservation are satisfied by construction at every training iteration.

\paragraph*{Summary}
This section established the exact equivalence between MERA-inspired layers and paraunitary filter banks. Section~\ref{subsec:poly-theory} introduced the polyphase theory framework, showing how two-channel filter banks operate in the decimated domain via constant polyphase matrices. Section~\ref{subsec:main-equivalence} proved the main result (Theorem~\ref{thm:circuit_pu}): MERA layers with orthogonal matrices $U_{\ell} \in \mathcal{O}(2)$ are mathematically equivalent to two-tap paraunitary filter banks, inheriting perfect reconstruction and energy conservation guarantees. Corollary~\ref{cor:qmf} showed that the Haar wavelet arises as the unique QMF filter maximizing DC gain. Section~\ref{subsec:variational} formulated the manifold-constrained learning objective, which promotes sparsity while enforcing orthogonality via polar projection at every training iteration. Section \ref{sec:learning-framework} presents the optimization algorithm; Section \ref{sec:experiments} validates performance on real network traces.

\section{Learning Framework}
\label{sec:learning-framework}

Having established in Section~\ref{sec:equivalence-PUFB} that MERA-inspired layers form paraunitary filter banks, the practical optimization pipeline is described next\footnote{The implementation is inspired by the MERA Julia code example released by Evenbly~\cite{TensorsNetMERA}.}. This pipeline learns the scale isometries $\{U_\ell\}_{\ell=1}^L$ directly from data while preserving PR, energy conservation, and numerical stability.

\subsection{Optimization Pipeline}
\label{subsec:pipeline}

The core learning procedure, detailed in Algorithm~\ref{alg:mera-opt}, implements a variational loop that optimizes the MERA-inspired filter banks on windowed traffic segments. The algorithm requires a real-valued signal window $\mathbf{x} \in \mathbb{R}^N$, the number of decomposition levels $L$, and non-negative loss weights $\lambda_{\text{sparse}}$ and $\lambda_{\text{MSE}}$. The output is a collection of scale isometries $\mathcal{U} = \{U_1, \ldots, U_L\}$ constrained to remain orthonormal throughout training.

\begin{algorithm}[t]
	\caption{MERA-Wavelet Optimization}
	\label{alg:mera-opt}
	\begin{algorithmic}[1]
		\Require $x$, $L$, $\texttt{numiter}$, $\eta$, $\lambda_{\text{sparse}}$, $\lambda_{\text{MSE}}$, $U_0$ (optional)
		\State $U \gets U_0$ if provided; otherwise initialize randomly and project onto $\mathcal{O}(2)$
		\For{$k = 1$ to $\texttt{numiter}$}
		\State $(a, \{d^{(\ell)}\}_{\ell=1}^L) \gets \textsc{MERAAnalyze}(x, U)$ \Comment{Forward transform}
		\State $\mathcal{L}_{\text{sparse}} \gets \frac{1}{n_d}\sum_{\ell=1}^L \|d^{(\ell)}\|_1$ \Comment{Mean $\ell_1$ norm}
		\If{$\lambda_{\text{MSE}} > 0$}
		\State $\hat{x} \gets \textsc{MERASynthesize}(a, \{d^{(\ell)}\}, U)$
		\State $\mathcal{L}_{\text{MSE}} \gets \frac{1}{N} \|\hat{x} - x\|_2^2$
		\Else
		\State $\mathcal{L}_{\text{MSE}} \gets 0$
		\EndIf
		\State $\mathcal{L} \gets \lambda_{\text{sparse}} \mathcal{L}_{\text{sparse}} + \lambda_{\text{MSE}} \mathcal{L}_{\text{MSE}}$
		\State $\nabla_U \gets \textsc{Backpropagate}(\mathcal{L}, U, x)$ \Comment{Euclidean gradient}
		\State $U \gets \textsc{AdamStep}(U, \nabla_U, \eta)$ \Comment{Update in $\mathbb{R}^{2 \times 2}$}
		\For{$\ell = 1$ to $L$}
		\State $U_\ell \gets U_\ell (U_\ell^\dagger U_\ell)^{-1/2}$ \Comment{Polar projection onto $\mathcal{O}(2)$}
		\EndFor
		\EndFor
		\State \Return $U$
	\end{algorithmic}
\end{algorithm}

Each iteration proceeds in three phases:
\begin{enumerate}
	\item \textbf{Forward analysis} (line~3): decompose $\mathbf{x}$ into multiscale coefficients $\{a^{(L)}, d^{(1)}, \ldots, d^{(L)}\}$.
	\item \textbf{Loss evaluation} (lines~4--11): compute the composite objective combining sparsity promotion and (optionally) reconstruction fidelity.
	\item \textbf{Constrained update} (lines~12--15): apply Adam gradient step followed by polar projection to restore orthogonality.
\end{enumerate}

\paragraph{Gradient flow and manifold projection}
Line~12 computes the Euclidean gradient $\nabla_U \mathcal{L}$ in the ambient space $\mathbb{R}^{2 \times 2}$ via automatic differentiation (AD). Line~13 applies Adam~\cite{Kingma2014} with learning rate $\eta$. This Euclidean step violates the orthogonality constraint $U_\ell^\dagger U_\ell = I$.

Line~15 restores the constraint via \emph{polar projection}: for each $U_\ell$, the nearest orthogonal matrix in Frobenius norm is $U_\ell (U_\ell^\dagger U_\ell)^{-1/2}$. By enforcing $U_\ell \in \mathcal{O}(2)$ at every iteration, the algorithm guarantees that paraunitarity, Parseval identity, and PR hold throughout training -- not merely as soft approximations.

\paragraph{Loss function}

The sparsity term $\mathcal{L}_{\text{sparse}}$ (line~4) promotes energy compaction by penalizing the mean absolute value of detail coefficients. The normalization factor $n_d$ denotes the total number of detail coefficients across all scales, ensuring that the effective weight of the sparsity penalty remains invariant with respect to the decomposition
depth $L$ and the signal length $N$.

\paragraph{Gradient-Based Optimization}
The framework employs AD to compute gradients $\nabla_{U} \mathcal{L}$ with respect to the filter parameters. Gradients are computed via the chain rule applied through the computational graph of the MERA transform -- a process commonly termed \emph{backpropagation} in the machine learning literature~\cite{Goodfellow-et-al-2016}. 

The Adam optimizer~\cite{Kingma2014} adapts learning rates per-parameter using exponential moving averages of first ($m$) and second gradient ($v$) moments:
\begin{align}
	m_t &= \beta_1 m_{t-1} + (1-\beta_1) \nabla_{U} \mathcal{L}, \\
	v_t &= \beta_2 v_{t-1} + (1-\beta_2) (\nabla_{U} \mathcal{L})^2, \\
	U_t &= U_{t-1} - \eta \frac{\hat{m}_t}{\sqrt{\hat{v}_t} + \epsilon},
\end{align}
where $\hat{m}_t$ and $\hat{v}_t$ are bias-corrected estimates. This adaptive scheme provides faster convergence and reduced sensitivity to hyperparameter selection compared to vanilla stochastic gradient descent. The default parameters ($\beta_1 = 0.9$, $\beta_2 = 0.999$, $\epsilon = 10^{-8}$) are used throughout (Table~\ref{tab:hyperparameters}).

\subsection{Initialization}
\label{subsec:init}

The framework adopts \emph{Haar warm-start initialization}: each $U_\ell$ is set to $U_{\text{Haar}} = \frac{1}{\sqrt{2}}\begin{bsmallmatrix}1 & 1 \\ 1 & -1\end{bsmallmatrix}$. This provides coarse energy compaction from the outset, and subsequent Adam updates refine this prior to match trace-specific correlations. Empirically, Haar initialization accelerates convergence compared to random starting points while achieving identical final performance.

For random initialization, each $U_\ell$ can be drawn from a Gaussian distribution and immediately projected onto $\mathcal{O}(2)$, providing a neutral baseline that does not bias the learned filters toward any wavelet family.

\subsection{Computational Complexity}
\label{subsec:complex}

The proposed framework preserves the efficiency of classical wavelet transforms:

\textbf{Inference.} Analysis and synthesis have complexity $\mathcal{O}(N)$, identical to the DWT. Each level $\ell$ processes $N/2^\ell$ samples with constant-cost $2 \times 2$ operations.

\textbf{Training.} Each iteration requires $\mathcal{O}(N)$ for forward/backward passes. Polar projection operates on $2 \times 2$ matrices with $\mathcal{O}(1)$ cost per level, contributing $\mathcal{O}(L) = \mathcal{O}(\log N)$ overhead. Total training cost is $\mathcal{O}(T \cdot N)$ for $T$ iterations.

\textbf{Parameter efficiency.} The learned transform requires only $4L$ scalar parameters (one $2 \times 2$ orthogonal matrix per level). For $L = 5$, this amounts to 20 trainable parameters -- enabling rapid adaptation without risk of overfitting.

\paragraph*{Summary}
This section presented the MERA-inspired wavelet learning framework. Algorithm~\ref{alg:mera-opt} integrates AD, Adam optimization, and manifold-constrained projection into a unified pipeline. Section~\ref{sec:experiments} validates the framework on six years of backbone traffic, demonstrating that learned orthonormal wavelets adapt to traffic-specific correlation structures while preserving the mathematical guarantees essential for 6G telemetry.


\section{Experimental Results}
\label{sec:experiments}

This section presents an empirical validation of the proposed adaptive MERA-inspired wavelet framework. The evaluation assesses whether learned orthonormal wavelets simultaneously achieve improved rate-distortion performance and preserve the LRD properties critical for DT synchronization.

\subsection{The LRD Preservation Requirement}
\label{sec:lrd_requirement}

As mentioned in Section~\ref{sec:intro}, backbone traffic exhibits LRD characterized by power-law autocorrelation decay:
\begin{equation}
	\phi(k) \sim k^{-\beta}, \quad 0 < \beta < 1,
	\label{eq:lrd}
\end{equation}
where the decay exponent $\beta$ relates to the Hurst parameter $H \in (0.5, 1)$ via $\beta = 2 - 2H$. This slow correlation decay has profound implications for network models commonly employed in DT frameworks:

\paragraph{Queueing Analysis} For finite buffers of size $B$ packets, the overflow probability under LRD input decays polynomially rather than 
exponentially \cite{Norros1995, Parulekar1996}:
\begin{equation}
	P(\text{overflow}) \sim B^{-(2-2H)} = B^{-\beta}.
	\label{eq:overflow_law}
\end{equation}
In contrast, Markovian (memoryless) models predict $P(\text{overflow}) \sim e^{-\lambda B}$. This disparity leads to $10^2$--$10^3\times$ errors in buffer dimensioning when $H$ is underestimated, fundamentally altering capacity provisioning rules for ultra-reliable low-latency communications (URLLC) in 6G systems.

\paragraph{Capacity Planning} The effective bandwidth required to meet target loss rates scales differently under LRD traffic compared to Poisson or exponential models \cite{Leland1994}. DT-driven optimization algorithms that rely on traffic statistics as input parameters will generate invalid provisioning decisions if the telemetry compression distorts $H$.

\paragraph{Consequence} If telemetry compression degrades the LRD signature (e.g., by attenuating long-timescale correlations), the DT's predictions become statistically inconsistent with the physical network. This can lead to mis-provisioning, SLA violations, or instability in closed-loop control scenarios.

\subsection{Experimental Setup}
\label{sec:setup}

\subsubsection{Dataset and Preprocessing}
The evaluation utilizes trans-Pacific backbone traces from the MAWI (Measurement and Analysis on the WIDE Internet) Working Group Traffic Archive~\cite{mawi_archive}. Six captures spanning 2020--2025 (Samplepoint-F) were selected to represent heterogeneous operating conditions, with traffic loads ranging from 314 Mbps to 1.75 Gbps. Packet-level metadata were aggregated into byte-per-millisecond time series. Table~\ref{tab:mawi_traces_summary} summarizes the characteristics of the traces.

At the time of this study, publicly available, large-scale, millisecond-resolution traffic traces from operational 5G or beyond-5G networks are not available. This limitation is widely acknowledged in the literature. Consequently, this work validates the proposed framework on backbone aggregation traces, which capture the emergent statistical properties -- particularly LRD -- that digital twins must preserve for stable closed-loop optimization.

\begin{table}[h]
	\centering
	\caption{MAWI trace characteristics (Samplepoint-F, 15-min captures).}
	\label{tab:mawi_traces_summary}
	\begin{tabular}{lccc}
		\toprule
		Trace & Duration & Packets & Avg.~rate \\
		\midrule
		\texttt{202004081229} & 900~s & 81~M & 314~Mbps \\
		\texttt{202103181400} & 900~s & 86~M & 416~Mbps \\
		\texttt{202204131100} & 900~s & 119~M & 769~Mbps \\
		\texttt{202301131400} & 900~s & 108~M & 776~Mbps \\
		\texttt{202406192000} & 900~s & 194~M & 1.75~Gbps \\
		\texttt{202504090300} & 900~s & 126~M & 885~Mbps \\
		\bottomrule
	\end{tabular}
\end{table}

\emph{Scope limitation:} These backbone traces capture aggregated traffic from thousands of sources, exhibiting the LRD structure characteristic of statistical multiplexing. The framework's performance on wireless edge telemetry -- where individual user dynamics, channel fading, and mobility introduce distinct correlation structures -- is deferred to future investigation (Section \ref{sec:conclusion}).

\begin{figure*}[!ht]
	\begin{subfigure}[t]{0.48\textwidth}
		\includegraphics[width=\linewidth]{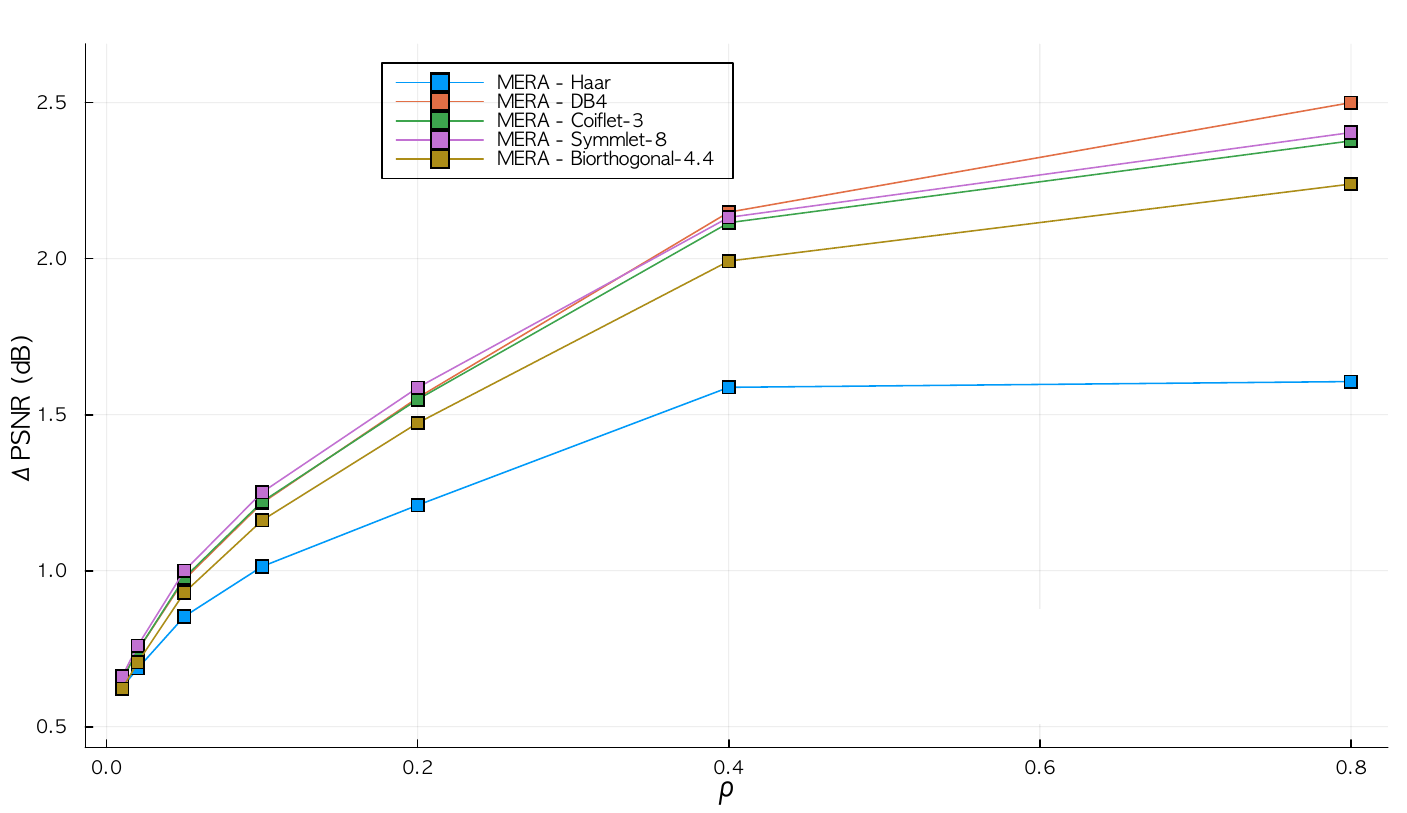}
		\caption{2020 (314 Mbps, $H{=}0.89$)}
		\label{fig:psnr_2020}
	\end{subfigure}
	\begin{subfigure}[t]{0.48\textwidth}
		\includegraphics[width=\linewidth]{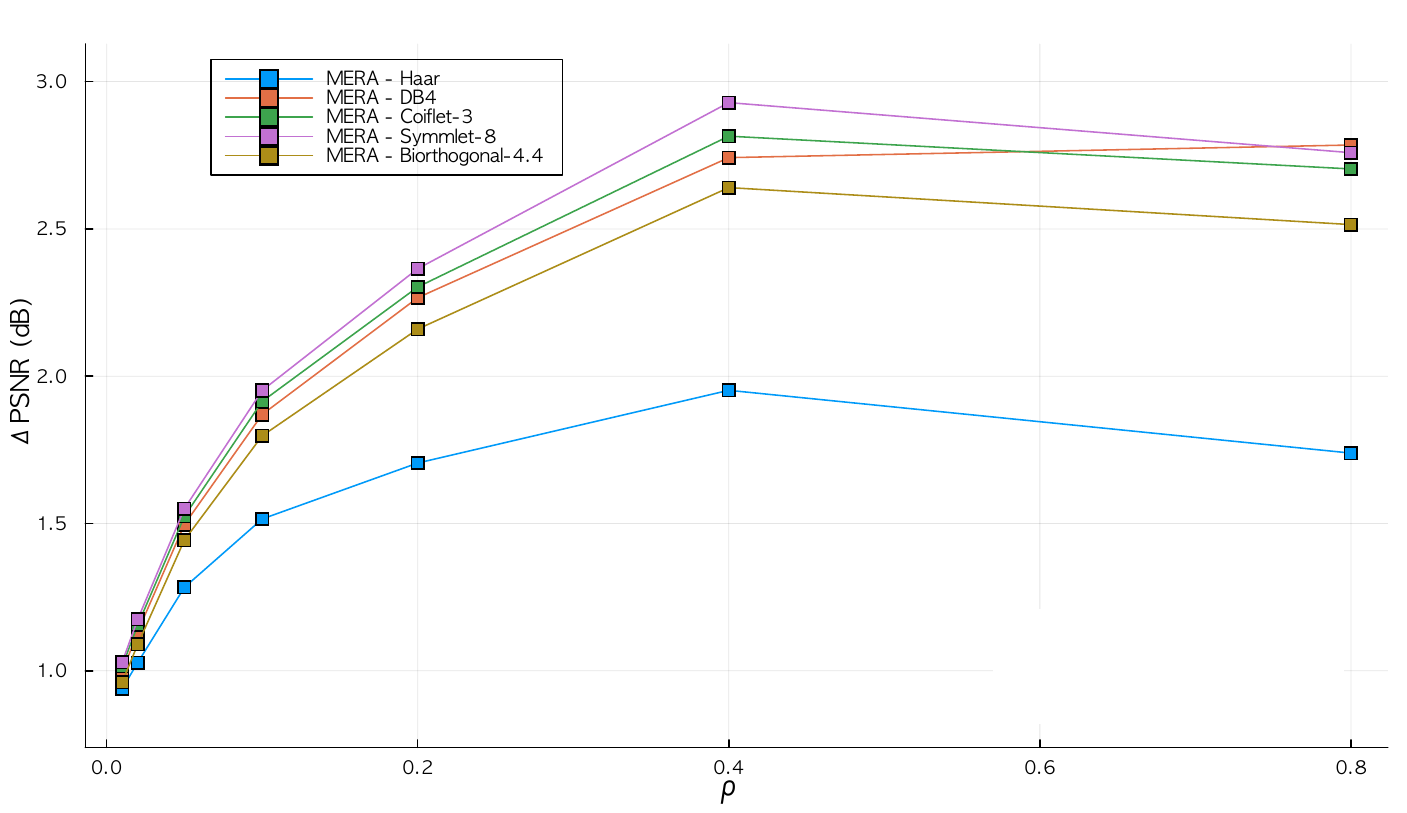}
		\caption{2021 (416 Mbps, $H{=}0.77$)}
		\label{fig:psnr_2021}
	\end{subfigure}
	
	\vspace{0.2cm}
	\begin{subfigure}[t]{0.48\textwidth}
		\includegraphics[width=\linewidth]{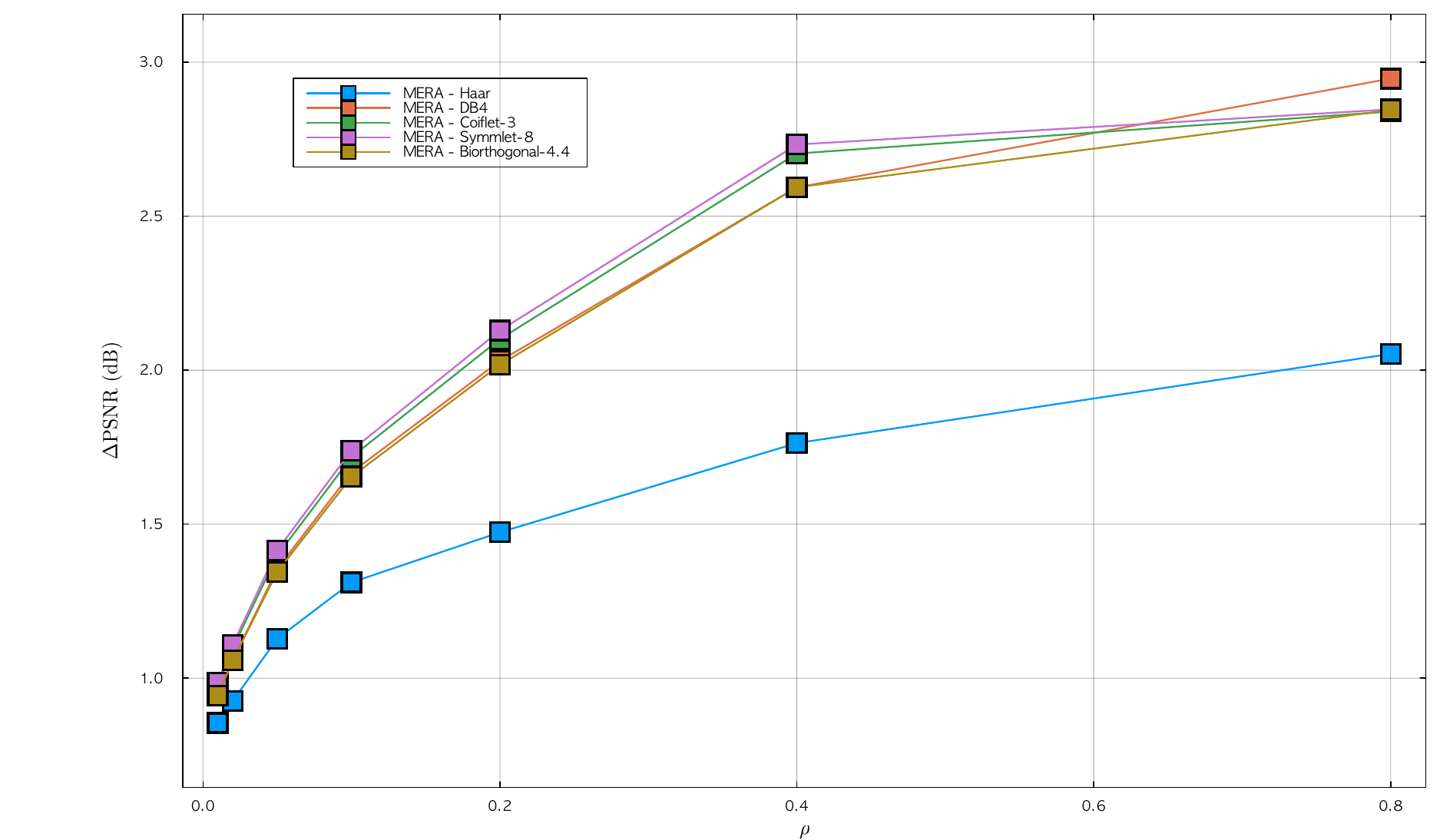}
		\caption{2022 (769 Mbps, $H{=}0.93$)}
		\label{fig:psnr_2022}
	\end{subfigure}
	\begin{subfigure}[t]{0.48\textwidth}
		\includegraphics[width=\linewidth]{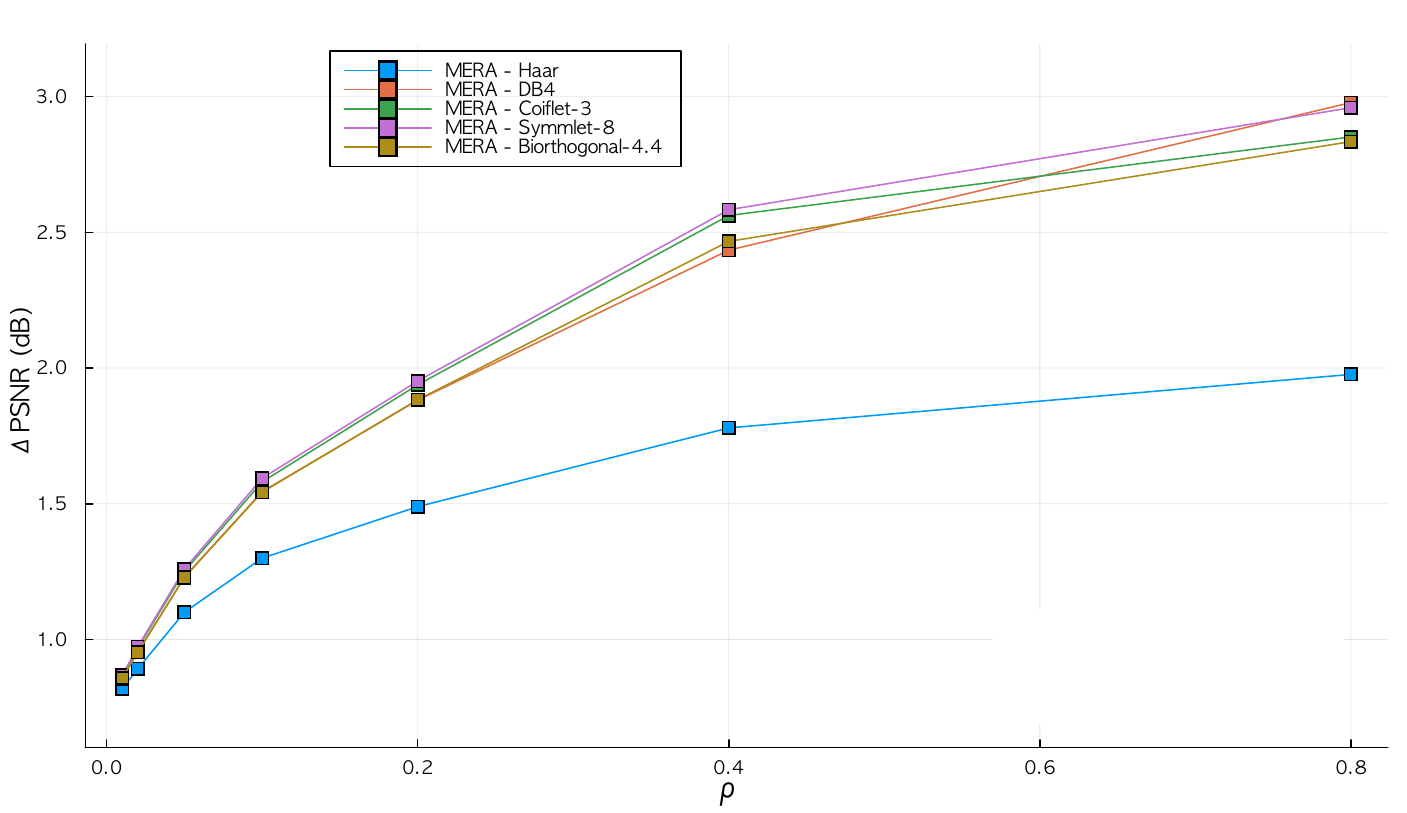}
		\caption{2023 (776 Mbps, $H{=}0.86$)}
		\label{fig:psnr_2023}
	\end{subfigure}\hfill
	
	\vspace{0.2cm}
	\begin{subfigure}[t]{0.48\textwidth}
		\includegraphics[width=\linewidth]{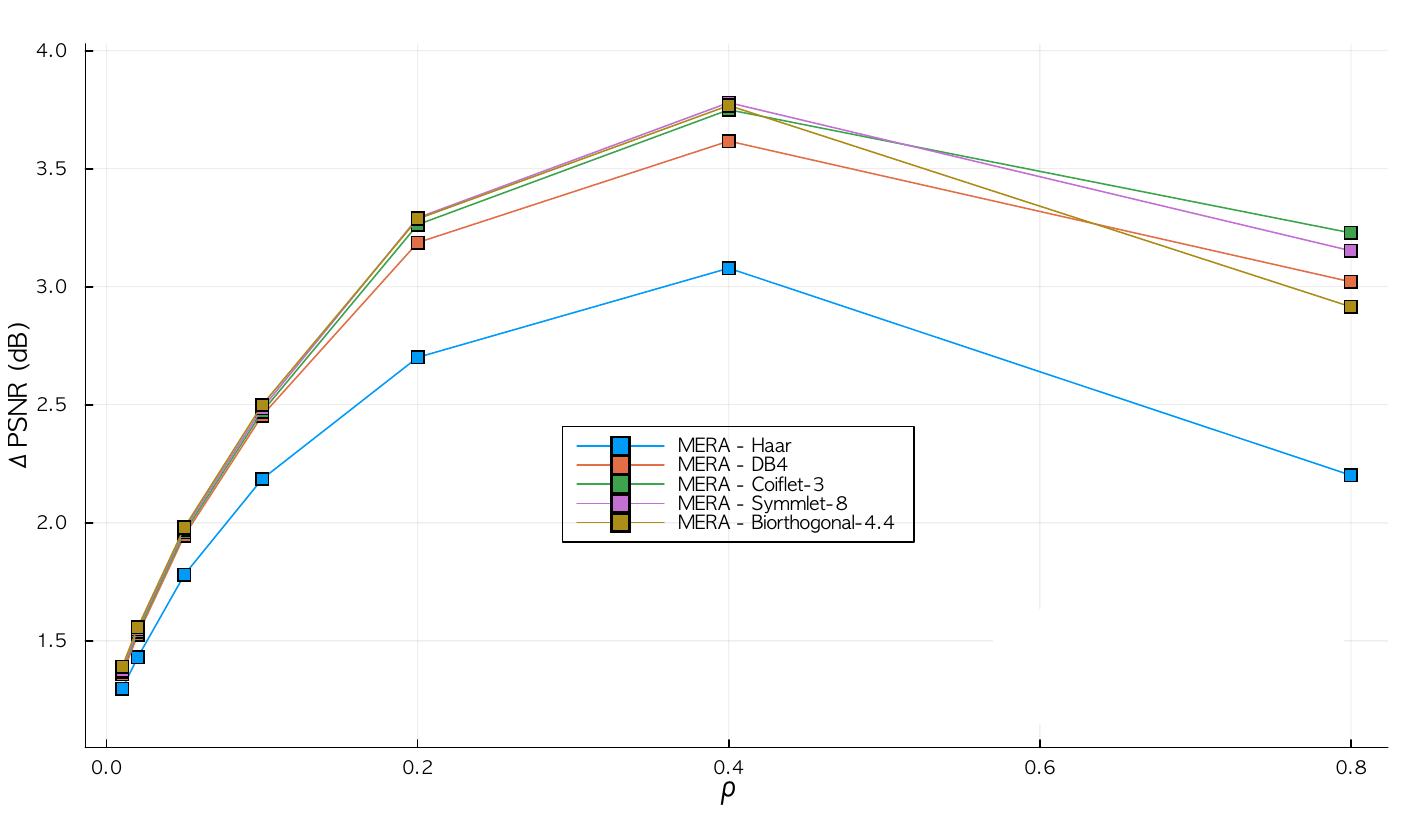}
		\caption{2024 (1.75 Gbps, $H{=}0.88$)}
		\label{fig:psnr_2024}
	\end{subfigure}
	\begin{subfigure}[t]{0.48\textwidth}
		\includegraphics[width=\linewidth]{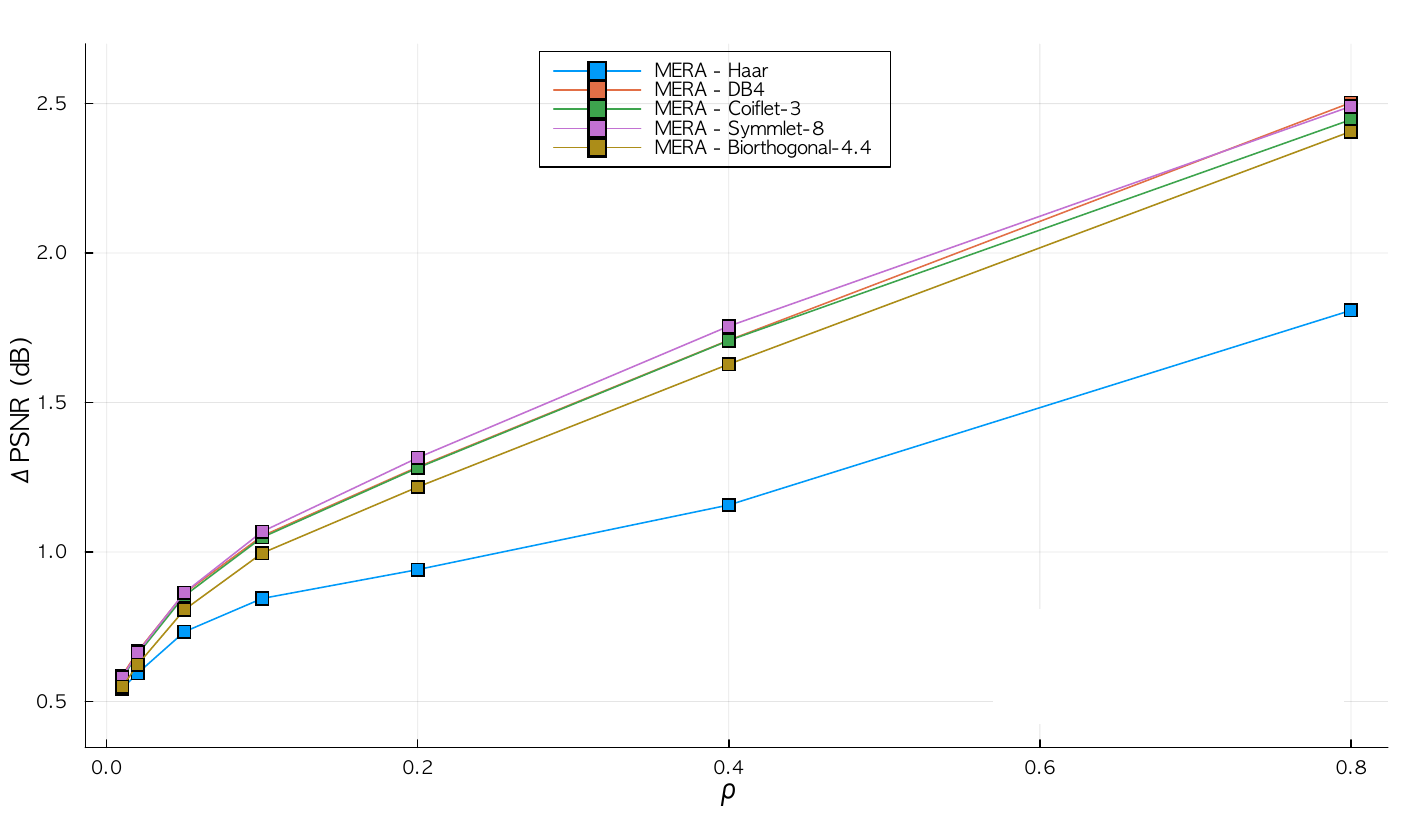}
		\caption{2025 (885 Mbps, $H{=}0.83$)}
		\label{fig:psnr_2025}
	\end{subfigure}
	
\caption{PSNR gains of MERA-learned wavelets over fixed baselines as a function of retention ratio $\rho$. A retention ratio of $\rho = 0.1$ corresponds to 90\% compression (retaining only 10\% of coefficients by magnitude). The learned filters consistently outperform classical wavelets across all compression levels and traffic conditions.}
\label{fig:psnr_all_traces}
\end{figure*}

\subsubsection{Training Configuration}
Experiments employ a two-stage training schedule optimizing MERA-wavelet parameters on 1024-sample non-overlapping windows. The optimization utilizes the Adam solver with a sparsity-driven objective ($\lambda_{\mathrm{sparse}} = 1.0, \lambda_{\mathrm{MSE}} = 0$), ensuring that the learned filters prioritize energy compaction into approximation coefficients. Table~\ref{tab:hyperparameters} details the complete hyperparameter configuration.

\begin{table*}[h]
	\centering
	\caption{MERA-Wavelet training hyperparameters for all experiments (Section~\ref{sec:experiments}).}
	\label{tab:hyperparameters}
	\begin{tabular}{llp{7cm}}
		\hline
		\textbf{Parameter} & \textbf{Value} & \textbf{Justification} \\
		\hline
		\multicolumn{3}{l}{\textit{Architecture}} \\
		Decomposition levels & $L=5$ & Captures scales $2^1$--$2^5$ (2--32 ms) \\
		Initialization & Haar warm-start & Leverages wavelet prior \\
		\hline
		\multicolumn{3}{l}{\textit{Optimization}} \\
		Total iterations & 100 (50 + 50) & Two-stage schedule \\
		Stage 1 learning rate & $\eta_1 = 5 \times 10^{-3}$ & Coarse adaptation \\
		Stage 2 learning rate & $\eta_2 = 2.5 \times 10^{-3}$ & Fine-tuning (halved $\eta_1$) \\
		Adam parameters & $\beta_1{=}0.9$, $\beta_2{=}0.999$ & Standard defaults \\
		Adam epsilon & $\epsilon = 10^{-8}$ & Numerical stability \\
		\hline
		\multicolumn{3}{l}{\textit{Loss function}} \\
		Sparsity weight & $\lambda_{\mathrm{sparse}} = 1.0$ & $\ell_1$ penalty on detail coefficients \\
		MSE weight & $\lambda_{\mathrm{MSE}} = 0.0$ & Disabled (no improvement observed) \\
		\hline
		\multicolumn{3}{l}{\textit{Data processing}} \\
		Window size & 1024 samples & Power-of-two for dyadic DWT \\
		Window stride & 1024 samples & Non-overlapping windows \\
		Retention ratios & $\rho \in \{0.01, \dots, 0.80\}$ & Rate-distortion evaluation \\
		\hline
		\multicolumn{3}{l}{\textit{Implementation}} \\
		Random seed & 12345 & Reproducibility \\
		Parametrization & MERA (polar proj.) & Algorithm \ref{alg:mera-opt}, Section~\ref{sec:learning-framework} \\
		Hardware & Apple M3 Pro & CPU-only execution \\
		\hline
	\end{tabular}
\end{table*}

\textit{Reproducibility:} All experiments ran on a single Apple M3 Pro laptop using Julia 1.11 with CPU-only execution. Random seed 12345 ensures deterministic initialization. The complete codebase, including hyperparameter configuration files, training scripts, and learned filters, is available at \url{https://github.com/alexandreblima/MERA-wavelets}.

\subsubsection{Baselines}
Performance is compared against fixed orthonormal wavelets: Haar (length-2), Daubechies-4 (\texttt{db4}), Coiflet-3, Symmlet-8, and Biorthogonal 4.4. These baselines allow us to isolate the benefits of data-driven adaptivity under strict paraunitary constraints.

\subsubsection{Evaluation Metrics}
\label{subsec:metrics}
Reconstruction fidelity is quantified using PSNR:
\begin{equation}
	\label{eq:psnr}
	\text{PSNR} = 10 \cdot \log_{10} \left( \frac{\text{MAX}_I^2}{\text{MSE}} \right), \quad \text{MSE} = \frac{1}{N} \sum_{i=1}^{N} (x_i - \hat{x}_i)^2
\end{equation}
where $\text{MAX}_I$ is the peak magnitude of the window. Statistical fidelity is assessed by the preservation of the Hurst exponent ($H$), estimated via Abry--Veitch wavelet regression~\cite{Abry1998}. The error metric is $\Delta H = H_{\text{compressed}} - H_{\text{orig}}$.

\subsection{Compression Performance}
\label{sec:compression_results}

After training, the learned filters are evaluated under varying bandwidth constraints by retaining only a fraction $\rho \in (0,1]$ of the wavelet coefficients ranked by magnitude. Specifically, given the full coefficient vector $\mathbf{c} = [\mathbf{a}^{(L)}, \mathbf{d}^{(L)}, \ldots, \mathbf{d}^{(1)}]$, the compressed representation retains the $\lceil \rho \cdot |\mathbf{c}| \rceil$ coefficients with largest absolute values, setting the remainder to zero. Reconstruction is then performed via the inverse MERA transform $\mathcal{S}_\theta$. 

This coefficient-thresholding approach follows standard practice in wavelet compression~\cite{Mallat2009} and enables direct comparison across retention ratios. Note that $\rho$ is an \emph{evaluation parameter} -- it does not appear in the training objective~\eqref{eq:loss_variational}. The learned filters are optimized for general sparsity (minimizing detail coefficient magnitudes), and the retention ratio is varied at test time to characterize rate-distortion performance across different compression levels.

To facilitate direct comparison with fixed wavelet baselines, performance is reported in terms of $\Delta\mathrm{PSNR}$, defined as
\begin{equation}
	\Delta\mathrm{PSNR}(\rho) \triangleq \mathrm{PSNR}_{\text{MERA}}(\rho) - \mathrm{PSNR}_{\text{baseline}}(\rho).
\end{equation}

Fig.~\ref{fig:psnr_all_traces} presents the Rate-Distortion performance for all six MAWI traces. The proposed MERA-inspired wavelet framework consistently outperforms fixed baselines across the full range of retention ratios ($\rho$).

\paragraph{Rate-Distortion Analysis}
The learned filters achieve PSNR gains ranging from $0.5$~dB to $3.8$~dB compared to the best fixed alternative.
\begin{itemize}
	\item \textbf{Peak Performance (2024):} The largest gains are observed in the 2024 trace (Fig.~\ref{fig:psnr_2024}), where MERA achieves a $3.8$~dB improvement over Coiflet-3, Symmlet-8, and Biorthogonal-4.4 . This trace corresponds to the highest network load (1.75 Gbps) and strong LRD ($H \approx 0.88$), indicating that adaptive filters effectively capture the bursty dynamics of saturated links.
	\item \textbf{Convergence of Baselines:} Higher-order fixed wavelets (\texttt{db4}, Coiflet, Symmlet) tend to cluster within a narrow performance band ($< 0.3$~dB difference). MERA breaks this ceiling, demonstrating that optimizing the spectral tilt of the filter bank yields benefits beyond simply increasing the number of vanishing moments.
	\item \textbf{Haar Comparison:} While Haar performs robustly due to its short support, MERA consistently surpasses it by $0.6$--$3.1$~dB, proving that the learned filters successfully balance time-domain localization with frequency selectivity.
\end{itemize}


\subsection{Statistical Fidelity (LRD Preservation)}
\label{sec:statistical_fidelity}

Beyond pointwise error (MSE), 6G DT will require the preservation of the self-similar traffic structure. Table~\ref{tab:hurst_global} lists the reference Global Hurst Exponents ($H$) for the raw traces, confirming persistent LRD ($0.77\le H \le 0.93$) across all years.

\begin{table}[h!]
	\centering
	\caption{Global Hurst exponent estimates from MAWI traces
		(Abry--Veitch regression, 95\% confidence interval).}
	\label{tab:hurst_global}
	\begin{tabular}{lcc}
		\toprule
		\textbf{Trace (MAWI)} & \textbf{$\hat H$} & \textbf{95\% CI} \\
		\midrule
		202004081229 & 0.8897 & [0.853, 0.926] \\
		202103181400 & 0.7674 & [0.684, 0.851] \\
		202204131100 & 0.9313 & [0.883, 0.979] \\
		202301131400 & 0.8641 & [0.817, 0.911] \\
		202406192000 & 0.8771 & [0.825, 0.929] \\
		202504090300 & 0.8329 & [0.787, 0.878] \\
		\bottomrule
	\end{tabular}
\end{table}

\paragraph{Hurst Exponent Preservation}
Table~\ref{tab:hurst_preservation} reports the deviation $\Delta H$ in the reconstructed signal. At a retention ratio of $\rho=0.1$ (90\% compression), the method maintains $|\Delta H| \le 0.03$ for all traces. This demonstrates that the learned basis functions preserve the power-law decay of the autocorrelation function even at high compression rates.

Importantly, increasing the retention factor $\rho$ does not necessarily improve the stability of the Hurst exponent. While larger $\rho$ preserves more coefficients, it also reintroduces small-amplitude detail components primarily associated with high-frequency fluctuations. Since Hurst exponent estimation depends on the stability of multiscale scaling behavior rather than local reconstruction fidelity, these weak high-frequency contributions may increase estimator sensitivity and perturb the slope of the wavelet logscale diagram. Conversely, moderate sparsification suppresses such weak detail coefficients, effectively acting as a structural denoising mechanism that stabilizes scaling statistics. Therefore, the observed variations of $\Delta H$ with $\rho$ reflect estimator sensitivity rather than degradation of the reconstructed signal.

The threshold $|\Delta H| \le 0.03$ was selected based on the statistical precision of the estimator. As derived from the 95\% confidence intervals for the raw traces (Table~\ref{tab:hurst_global}), the intrinsic uncertainty of the Abry--Veitch estimator for these finite-length windows ranges from $\pm 0.036$ (Trace 2020) to $\pm 0.083$ (Trace 2021). Even for the critical high-load scenario (Trace 2024), the measurement error is $\approx \pm 0.052$. Consequently, maintaining compression deviations within $0.03$ ensures that the LRD structure of the reconstructed telemetry remains statistically indistinguishable from the original source, preserving the validity of the data for queueing analysis within the limits of measurement precision.

The spectral analysis in Fig.~\ref{fig:wavelet_spectra} corroborates this, showing that the energy distribution across scales $\ell$ maintains linearity. While deviations occur at very large scales ($\ell > 15$) due to finite-size effects and non-stationarity, the primary scaling region essential for LRD modeling is preserved.

\begin{table}[!htbp]
	\centering
	\caption{Hurst exponent deviations across compression levels (learned MERA).
		Values with $|\Delta H| \leq 0.03$ are shown in \textbf{bold}, indicating strong preservation of LRD.}
	\label{tab:hurst_preservation}
	\small
	\begin{tabular}{@{}lccccc@{}}
		\toprule
		Trace & $H_{\text{orig}}$
		& $\Delta H(\rho{=}0.1)$ 
		& $\Delta H(\rho{=}0.2)$ 
		& $\Delta H(\rho{=}0.4)$ 
		& $\Delta H(\rho{=}0.8)$ \\
		\midrule
		2020 & 0.890 
		& \textbf{+0.027} & \textbf{-0.011} & -0.038 & -0.049 \\
		2021 & 0.767 
		& \textbf{+0.011} & \textbf{-0.017} & -0.039 & -0.051 \\
		2022 & 0.931 
		& \textbf{+0.002} & \textbf{-0.028} & -0.055 & -0.064 \\
		2023 & 0.864 
		& \textbf{+0.020} & \textbf{-0.012} & -0.039 & -0.048 \\
		2024 & 0.877 
		& \textbf{-0.010} & -0.041 & -0.065 & -0.079 \\
		2025 & 0.833 
		& \textbf{-0.004} & -0.046 & -0.078 & -0.086 \\
		\midrule
		Mean & 0.846 
		& \textbf{+0.009} & \textbf{-0.026} & -0.052 & -0.063 \\
		\bottomrule
	\end{tabular}
\end{table}

\begin{figure*}[!ht]
	\centering
	\includegraphics[width=0.9\linewidth]{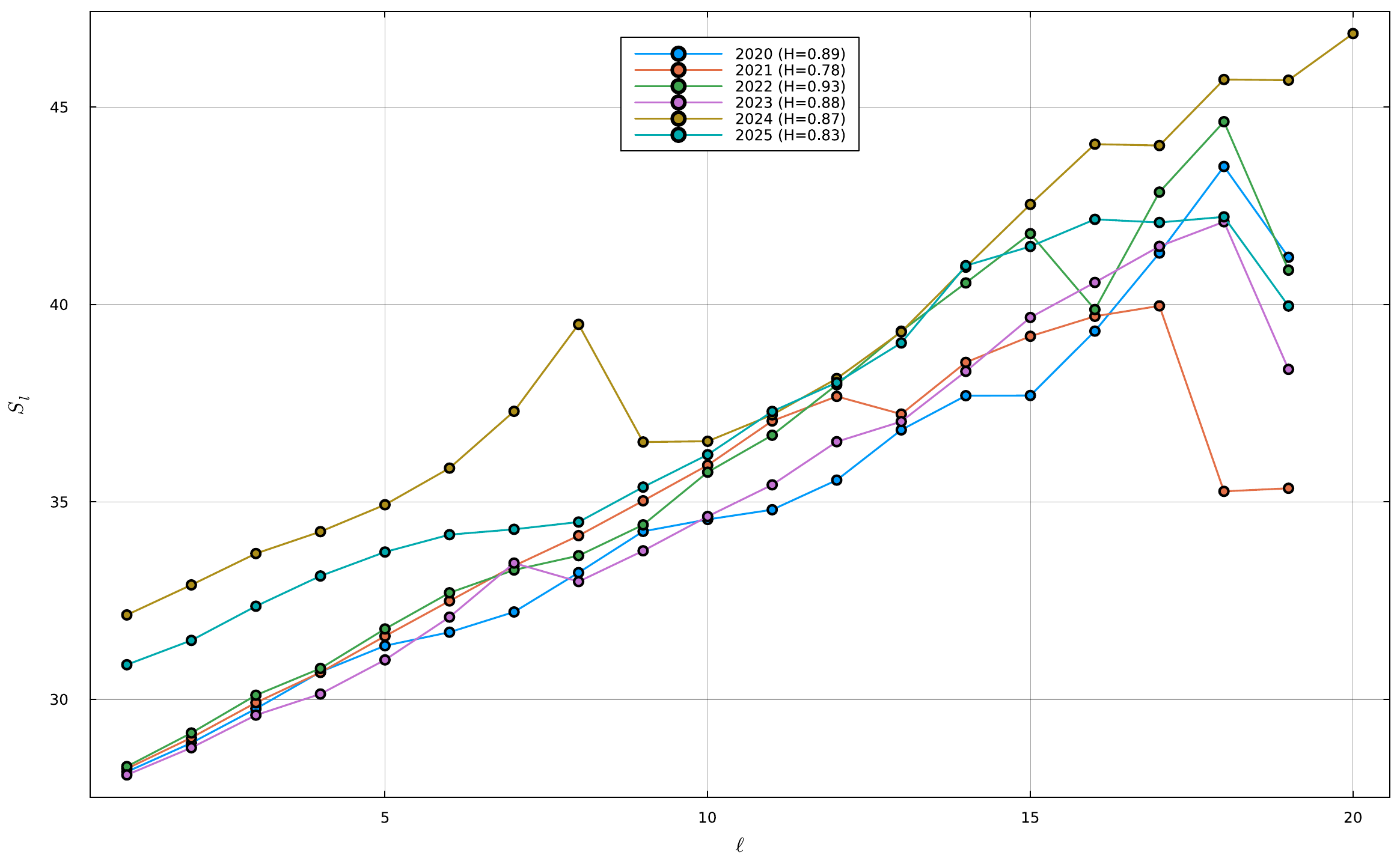}
	\caption{Wavelet energy spectra $S_\ell$ vs. scale $\ell$ for MAWI traces. The linear growth confirms power-law scaling (LRD). MERA filters are optimized to match this spectral tilt.}
	\label{fig:wavelet_spectra}
\end{figure*}

\subsection{Learned Filter Analysis}
\label{sec:filter_analysis}

To understand the adaptation mechanism, we examine the filters learned from the 2024 trace (Fig.~\ref{fig:learned_filters}). Starting from a Haar initialization, the optimization converges to an asymmetric structure (Table~\ref{tab:learned_coeffs}) that increases the support of the basis functions.

The frequency response reveals that the learned filters introduce specific passband ripples that deviate from the ``maximum flatness'' criteria of Daubechies wavelets. These deviations are not artifacts but data-driven adaptations that maximize energy compaction for the specific spectral signature of internet traffic, validating the use of the MERA framework for discovering domain-specific orthogonal bases.


\begin{figure*}[!ht]
	\centering
	\includegraphics[width=0.9\textwidth]{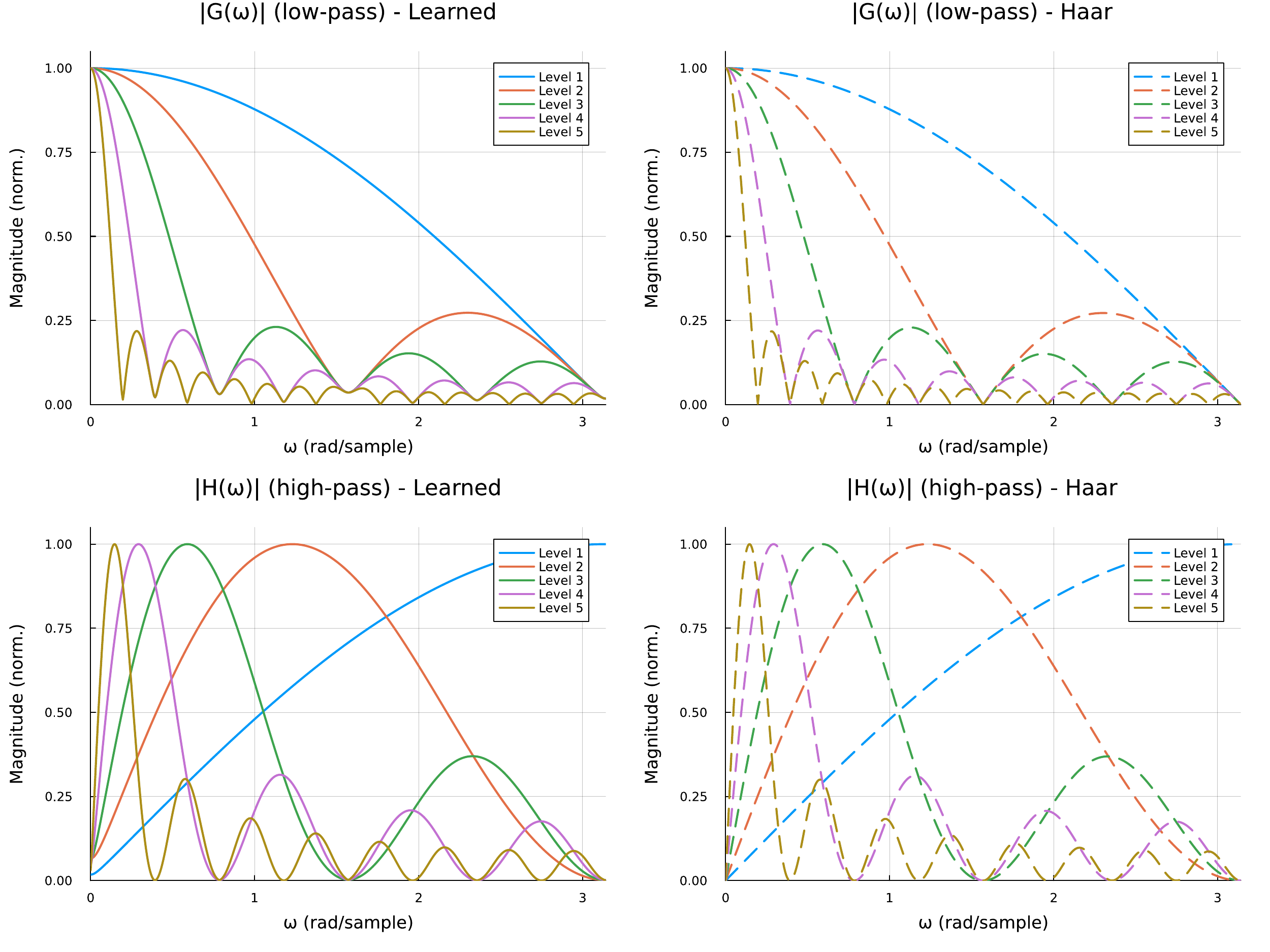}
	\caption{Frequency response of analysis filters (Trace~2024). \textbf{Left:} learned paraunitary filters. \textbf{Right:} Haar initialization. 
		\textbf{Top:} low-pass cascades $G_\ell(\omega)$. \textbf{Bottom:} high-pass filters $H_\ell(\omega)$. 
		Haar itself exhibits ripple due to the length‑2 basis. Both learned cascades keep this ripple structure; however, the low-pass shift their amplitude and zero locations slightly across levels, especially in the pass‑band and transition regions.}
	\label{fig:learned_filters}	
\end{figure*}

\begin{table*}[!ht]
	\centering
	\caption{Learned filter coefficients for trace \texttt{202406192000}. Deviations from Haar (length-2) indicate adaptation to traffic structure.}
	\label{tab:learned_coeffs}
	\small
	\begin{tabular}{@{}ccccc@{}}
		\toprule
		Level $\ell$ & $\|g_\ell - g_{\text{Haar}}\|_2$ & $g_\ell$ (low-pass) & $h_\ell$ (high-pass) \\
		\midrule
		1 & 0.0177 & $[0.7195,\,0.6945]^\top$ & $[\phantom{-}0.6945, -0.7195]^\top$ \\
		2 & 0.0505 & $[0.7419,\,0.6705]^\top$ & $[\phantom{-}0.6705, -0.7419]^\top$ \\
		3 & 0.0473 & $[0.7398,\,0.6729]^\top$ & $[\phantom{-}0.6729, -0.7398]^\top$ \\
		4 & 0.0331 & $[0.6833,\,0.7301]^\top$ & $[\phantom{-}0.7301, -0.6833]^\top$ \\
		5 & 0.0222 & $[0.7226,\,0.6913]^\top$ & $[\phantom{-}0.6913, -0.7226]^\top$ \\
		\bottomrule
	\end{tabular}
\end{table*}

\paragraph*{Summary}
The experimental validation confirms that the proposed MERA-wavelet framework effectively bridges the gap between theoretical orthogonality and data-driven adaptation. The key takeaways for 6G DT implementations are:

\begin{itemize}
	\item \textbf{Rate-Distortion Superiority:} The learned filters achieve consistent PSNR gains of $0.5$--$3.8$\,dB over standard wavelet families. The advantage is most pronounced in high-load, bursty scenarios (e.g., the 2024 trace with 1.75 Gbps), confirming that adaptive bases successfully capture the non-stationary dynamics of modern backbone traffic.
	
	\item \textbf{Statistical Preservation:} Crucially for predictive modeling, the method preserves the self-similar nature of the traffic. The Hurst exponent deviations remain negligible ($|\Delta H| \le 0.03$) even at 90\% compression ($\rho=0.1$), ensuring that the reconstructed telemetry retains the correlation structure necessary for accurate network simulation.
	
	\item \textbf{Structural Guarantees:} Unlike unconstrained deep learning approaches, the MERA-based optimization converges to interpretable, perfectly reconstructing filter banks. The results demonstrate that strict paraunitary constraints can be maintained without sacrificing the flexibility required to adapt to diverse spectral signatures.
\end{itemize}

These findings position the MERA-wavelet not merely as a compression tool, but as a reliable interface for high-fidelity data synchronization in 6G architectures.


\section{Conclusion}
\label{sec:conclusion}

This work addressed a central challenge in the design of DT for 6G networks: achieving high-fidelity telemetry compression while preserving the strict structural guarantees required for reliable closed-loop operation. Instead of treating compression as a generic rate -- distortion problem, the proposed approach framed telemetry as a synchronization mechanism, where violations of invertibility, energy conservation, or LRD directly compromise the predictive stability of the DT.

A rigorous and exact equivalence between MERA tensor networks and two-channel paraunitary wavelet filter banks was established, enabling a learning framework that overcomes the limitations of fixed wavelet designs. Experimental validation on real-world backbone aggregation traces demonstrated consistent rate -- distortion gains of up to 3.8~dB over classical orthogonal and biorthogonal wavelets, while preserving the self-similar structure of traffic within strict Hurst-exponent bounds. These improvements were obtained without relaxing paraunitary constraints, ensuring PR and Parseval energy conservation at all scales.

Beyond compression performance, the results position the MERA-wavelet framework as a principled synchronization interface between physical networks and their DT. By preserving the multiscale statistical invariants that underpin traffic modeling, the proposed method provides a technology-agnostic foundation for telemetry pipelines in bandwidth-constrained 6G architectures. Extensions to wireless and edge environments, where mobility and radio-induced nonstationarity introduce additional challenges, constitute a natural direction for future investigation.

\appendix

\section{Proof of Theorem~\ref{thm:circuit_pu} (Architectural Equivalence)}
\label{appendix:proof_equivalence}
	
\textbf{Theorem~\ref{thm:circuit_pu} (Architectural Equivalence).}
A MERA-inspired layer (Definition~\ref{def:mera_layer}) is equivalent to a two-channel paraunitary filter bank whose polyphase representation is a constant orthonormal matrix $\mathbf{E}(z) \equiv U_\ell$:
\begin{equation}
	\begin{bmatrix} A(z) \\ D(z) \end{bmatrix}
	=
	\begin{bmatrix} g_0 & g_1 \\ h_0 & h_1 \end{bmatrix}
	\begin{bmatrix} X_0(z) \\ X_1(z) \end{bmatrix}.
\end{equation}
	
\begin{proof}[Proof of Theorem~\ref{thm:circuit_pu}]
Both directions of the equivalence are derived.
		
		
\textbf{(Sufficiency)}  
Assume a MERA-inspired local operator $U_\ell \in \mathcal{O}(2)$ acts pointwise on adjacent pairs of samples,
\begin{equation}
		\begin{bmatrix} a_k \\ d_k \end{bmatrix} 
		= U_\ell 
		\begin{bmatrix} x_{2k} \\ x_{2k+1} \end{bmatrix} 
		= 
		\begin{bmatrix} g_0 & g_1 \\ h_0 & h_1 \end{bmatrix}
		\begin{bmatrix} x_{2k} \\ x_{2k+1} \end{bmatrix},
		\quad k \in \mathbb{Z}.
		\label{eq:mera_time}
\end{equation}
It is shown next that this local transformation induces a paraunitary filter bank.
		
\textbf{Step 1 (Time-domain output):}
Expanding \eqref{eq:mera_time} yields the component-wise relations
	\begin{equation}
		a_k = g_0 x_{2k} + g_1 x_{2k+1}, 
		\qquad 
		d_k = h_0 x_{2k} + h_1 x_{2k+1}.
	\end{equation}
		
\textbf{Step 2 (Polyphase decomposition):}
Define the decimated z-transforms of the outputs
\begin{equation}
		A(z) = \sum_k a_k z^{-k}, 
		\qquad
		D(z) = \sum_k d_k z^{-k},
\end{equation}
and the even/odd polyphase components of the input
\begin{equation}
		X_0(z) = \sum_k x_{2k}z^{-k}, 
		\qquad 
		X_1(z) = \sum_k x_{2k+1}z^{-k}.
\end{equation}
		
\textbf{Step 3 (Z-transform substitution):}
Substituting the expressions for $a_k$ and $d_k$ from Step 1 into the z-transforms gives
\begin{align}
		A(z) &= g_0 X_0(z) + g_1 X_1(z), \\
		D(z) &= h_0 X_0(z) + h_1 X_1(z),
\end{align}
which can be written compactly in matrix form as
\begin{equation}
		\begin{bmatrix} A(z) \\ D(z) \end{bmatrix} 
		= 
		\begin{bmatrix} g_0 & g_1 \\ h_0 & h_1 \end{bmatrix}
		\begin{bmatrix} X_0(z) \\ X_1(z) \end{bmatrix}
		= \mathbf{E}(z) \begin{bmatrix} X_0(z) \\ X_1(z) \end{bmatrix},
\end{equation}
where $\mathbf{E}(z) \equiv U_\ell$ is the constant polyphase matrix.
		
\textbf{Step 4 (Two-tap FIR filters):}
Applying \eqref{eq:poly1}--\eqref{eq:poly2}, the analysis filters are
\begin{equation}
\begin{aligned}
		G(z) &= E_{00}(z^2) + z^{-1}E_{01}(z^2), \\
		H(z) &= E_{10}(z^2) + z^{-1}E_{11}(z^2).
\end{aligned}
\label{eq1:polyphase-step4}
\end{equation}
Since $\mathbf{E}(z)\equiv U_\ell$ is constant (z-independent), substituting the scalar entries  $E_{00}=g_0$, $E_{01}=g_1$, $E_{10}=h_0$, $E_{11}=h_1$ yields
\begin{equation}
		G(z) = g_0 + g_1 z^{-1}, 
		\qquad 
		H(z) = h_0 + h_1 z^{-1},
		\label{eq:fir_filters}
\end{equation}
which are length-2 FIR analysis filters parameterized by the entries of $U_\ell$.
		
\textbf{Step 5 (Paraunitarity):}
The orthogonality condition $U_\ell \in \mathcal{O}(2)$ directly implies paraunitarity of the polyphase matrix:
\begin{equation}
		\mathbf{E}(z)\,\mathbf{E}^{\adj}(z^{-1}) 
		= 
		U_\ell\,{U_\ell}^{\adj} 
		= I.
\end{equation}
This ensures power complementarity in the frequency domain: $|G(\omega)|^2 + |H(\omega)|^2 = 2$.
		
\textbf{Step 6 (Perfect reconstruction):}
Choosing the synthesis polyphase matrix $\mathbf{R}(z) = \mathbf{E}^{\adj}(z^{-1}) = {U_\ell}^{\adj}$ ensures  $\mathbf{R}(z)\mathbf{E}(z) = I$, guaranteeing alias cancellation. In the time domain, this yields
\begin{equation}
		{U_\ell}^{\adj}\!\left(
		U_\ell 
		\begin{bmatrix} x_{2k} \\ x_{2k+1} \end{bmatrix}
		\right)
		=
		({U_\ell}^{\adj}U_\ell)
		\begin{bmatrix} x_{2k} \\ x_{2k+1} \end{bmatrix}
		=
		\begin{bmatrix} x_{2k} \\ x_{2k+1} \end{bmatrix},
\end{equation}
confirming perfect reconstruction of the input samples.
		
\textbf{Conclusion:}
A MERA layer with $U_\ell\in\mathcal{O}(2)$ induces a critically sampled, two-channel, two-tap paraunitary filter bank with constant polyphase matrix $\mathbf{E}(z)\equiv U_\ell$, inheriting guarantees such as perfect reconstruction, energy conservation (Parseval identity), and $O(N)$ complexity. \qed
		
		
\medskip
\textbf{(Necessity)}  
Suppose now that the analysis stage of a two-channel paraunitary filter bank has a constant polyphase matrix
\begin{equation}
		\mathbf{E}(z) \equiv U = 
		\begin{bmatrix} g_0 & g_1 \\ h_0 & h_1 \end{bmatrix}, 
		\qquad
		U\,U^{\adj} = I.
\end{equation}
It is established next that this filter bank necessarily implements a MERA-inspired layer.
		
\textbf{Step 1 (Polyphase representation):} 
By the polyphase decomposition of the two-channel analysis bank, the output transforms are
\begin{equation}
		\begin{bmatrix} A(z) \\ D(z) \end{bmatrix} 
		= 
		\mathbf{E}(z) \begin{bmatrix} X_0(z) \\ X_1(z) \end{bmatrix},
\end{equation}
where $X_0(z) = \sum_k x_{2k} z^{-k}$ and $X_1(z) = \sum_k x_{2k+1} z^{-k}$ are the even and odd polyphase components of the input signal.
		
\textbf{Step 2 (Time-domain relation):} 
Since $\mathbf{E}(z) \equiv U$ is constant (z-independent), all polyphase entries are scalars. Matching coefficients in the z-transform yields the time-domain relation
\begin{equation}
		\begin{bmatrix} a_k \\ d_k \end{bmatrix}
		= 
		U \begin{bmatrix} x_{2k} \\ x_{2k+1} \end{bmatrix}, 
		\qquad k \in \mathbb{Z}.
\end{equation}
This shows that the analysis operation applies the \emph{same} matrix $U$ to each pair of adjacent samples $(x_{2k}, x_{2k+1})$ independently, followed by implicit downsampling.
		
\textbf{Step 3 (Equivalence to MERA layer):} 
The pairwise transformation in the previous step is precisely the definition of a MERA-inspired layer (Definition~\ref{def:mera_layer}):
\begin{equation}
		\begin{bmatrix} a_k \\ d_k \end{bmatrix}
		= 
		U_\ell \begin{bmatrix} x_{2k} \\ x_{2k+1} \end{bmatrix}.
\end{equation}
Thus, the filter bank analysis coincides exactly with the action of a MERA layer using $U_\ell = U$.
		
\textbf{Step 4 (Paraunitarity verification):} 
The paraunitarity condition $\mathbf{E}(z)\mathbf{E}^{\adj}(z^{-1}) = I$ reduces to $U U^{\adj} = I$ for constant $\mathbf{E}(z) \equiv U$, confirming that $U \in \mathcal{O}(2)$. This ensures perfect reconstruction via $U^{\adj}$ and energy conservation (Parseval identity).
		
\textbf{Step 5 (Two-tap FIR structure):} 
Applying \eqref{eq:poly1}--\eqref{eq:poly2} to the constant polyphase matrix yields the analysis filters
\begin{equation}
		G(z) = E_{00}(z^2) + z^{-1}E_{01}(z^2) = g_0 + g_1 z^{-1}, 
\end{equation}
\begin{equation}
		H(z) = E_{10}(z^2) + z^{-1}E_{11}(z^2) = h_0 + h_1 z^{-1},
\end{equation}
which are two-tap FIR filters parameterized by the entries of $U$.
		
\textbf{Conclusion:} 
Any two-channel paraunitary filter bank with constant polyphase matrix $\mathbf{E}(z) \equiv U \in \mathcal{O}(2)$ necessarily implements a MERA-inspired layer with two-tap FIR analysis filters. This completes the proof of equivalence. \qed

\end{proof}
	
	
\section{Proof of Corollary~\ref{cor:qmf} (Uniqueness of Haar for Two-Tap QMF)}
\label{appendix:proof_qmf}
	
\begin{proof}[Proof of Corollary~\ref{cor:qmf}]
It is shown that the Haar wavelet is the unique real two-tap FIR filter bank satisfying both PR and QMF paraunitarity.
		
\textbf{Step 1 (Orthogonality):}
The paraunitarity condition $U_\ell U_\ell^{\adj} = I$ applied to \eqref{eq:qmf_polyphase} yields
\begin{equation}
\begin{bmatrix} 
		g_0 & g_1 \\ 
		g_1 & -g_0 
\end{bmatrix}
\begin{bmatrix} 
		g_0 & g_1 \\ 
		g_1 & -g_0 
\end{bmatrix}^{\adj}
		=
\begin{bmatrix} 
		g_0^2 + g_1^2 & 0 \\ 
		0 & g_0^2 + g_1^2 
\end{bmatrix}
		= I.
\end{equation}
This immediately gives the normalization constraint
\begin{equation}
		g_0^2 + g_1^2 = 1.
\label{eq:norm_constraint}
\end{equation}
		
\textbf{Step 2 (Parameterization):}
Eq. \eqref{eq:norm_constraint} parameterizes all solutions as points on the unit circle:
\begin{equation}
	   g_0 = \cos\theta, \qquad g_1 = \sin\theta, \qquad \theta \in [0, 2\pi).
\end{equation}
		
\textbf{Step 3 (DC response maximization):}
Among all orthonormal solutions, the Haar wavelet uniquely maximizes the DC response $|G(0)|$:
\begin{equation}
		|G(0)| = |g_0 + g_1| = |\cos\theta + \sin\theta|.
\end{equation}
This is maximized when $\cos\theta = \sin\theta$, i.e., $\theta = \pi/4$, yielding
\begin{equation}
		g_0 = g_1 = \frac{1}{\sqrt{2}},
		\qquad
		|G(0)| = \sqrt{2}.
\end{equation}
		
\textbf{Step 4 (Uniqueness):}
Combining orthogonality \eqref{eq:norm_constraint} with the symmetry requirement $g_0 = g_1$ gives the unique solution
\begin{equation}
		g_0 = g_1 = \frac{1}{\sqrt{2}},
		\qquad
		h_0 = g_1 = \frac{1}{\sqrt{2}},
		\qquad
		h_1 = -g_0 = -\frac{1}{\sqrt{2}},
\end{equation}
corresponding to the Haar filters
\begin{equation}
		G(z) = \frac{1}{\sqrt{2}}(1 + z^{-1}), 
		\qquad
		H(z) = \frac{1}{\sqrt{2}}(1 - z^{-1}),
\end{equation}
with polyphase matrix
\begin{equation}
		U_{\text{Haar}} = \frac{1}{\sqrt{2}}
		\begin{bmatrix} 
			1 & 1 \\ 
			1 & -1 
		\end{bmatrix}.
\end{equation}
		
\textbf{Conclusion:}
Therefore, for two-tap filters, the QMF-paraunitary family forms a one-parameter manifold $\mathcal{M} = \{U(\theta) : \theta \in [0, 2\pi)\}$ with
\[
		U(\theta) = \begin{pmatrix} \cos\theta & \sin\theta \\ 
			\sin\theta & -\cos\theta \end{pmatrix}.
\]
The Haar filter bank corresponds to $\theta = \pi/4$, yielding the computationally simplest coefficients $g_0 = g_1 = 1/\sqrt{2}$ and uniquely maximizing the DC gain $|G(0)| = \sqrt{2}$ among all members of $\mathcal{M}$. This makes Haar the canonical choice for initialization, while the learnable angles $\theta_\ell$ explored in this work span the full QMF-paraunitary family.
\end{proof}
	
\textbf{Remark (Relationship to QMF).}  
For two-tap filters, the QMF-paraunitary family forms a one-parameter manifold within the reflection component of $\mathcal{O}(2)$ (i.e., $\det(U) = -1$). With Haar initialization ($\theta_\ell = \pi/4$) and polar projection, the learned filters remain in this QMF-paraunitary family throughout training. The PSNR gains in Figs.~\ref{fig:psnr_2020}--\ref{fig:psnr_2025} arise from learning optimal rotation angles $\theta_\ell \neq \pi/4$ that better match trace-specific LRD statistics, while preserving both QMF structure and perfect reconstruction. Extending to rotations ($\det = +1$) would exit the QMF family; longer filters ($N > 2$) would enlarge the design space further.

\bibliographystyle{IEEEtran}

\end{document}